\pdfoutput=1
\documentclass[11pt,oneside,a4paper]{article}
\usepackage[T1]{fontenc}
\usepackage[utf8]{inputenc}
\usepackage[english]{babel}
\usepackage{fullpage}
\usepackage{amsmath,amssymb}
\usepackage{amsthm}
\usepackage{mathtools}
\usepackage{xcolor}
\usepackage{enumitem}
\usepackage{natbib}
\usepackage{booktabs}
\usepackage{graphicx}
\usepackage{tikz}
\usepackage{pgfplots}
\pgfplotsset{compat=1.18}

\newcommand{\cF}{\mathcal{F}}

\newcommand{\E}{\mathbb{E}}

\newcommand{\I}{\mathbb{I}}

\newcommand{\K}{\mathbb{K}}

\newcommand{\Pb}{\mathbb{P}}
\newcommand{\Q}{\mathbb{Q}}
\newcommand{\R}{\mathbb{R}}

 \newcommand{\e}{\varepsilon}

\newcommand{\lrb}[1]{\left(#1\right)}
\newcommand{\brb}[1]{\bigl(#1\bigr)}

\newcommand{\lsb}[1]{\left[#1\right]}
\newcommand{\bsb}[1]{\bigl[#1\bigr]}

\newcommand{\lcb}[1]{\left\{#1\right\}}
\newcommand{\bcb}[1]{\bigl\{#1\bigr\}}

\newcommand{\labs}[1]{\left\lvert#1\right\rvert}
\newcommand{\babs}[1]{\bigl\lvert#1\bigr\rvert}

\newcommand{\lno}[1]{\left\lVert#1\right\rVert}

\newcommand{\lan}[1]{\left\langle#1\right\rangle}

\DeclareMathOperator*{\argmax}{argmax}
\newcommand{\diff}{\,\mathrm{d}}
\newcommand{\dif}{\mathrm{d}}

\newcommand{\fracc}[2]{#1/#2}

\theoremstyle{plain}
\newtheorem{theorem}{Theorem}
\newtheorem{lemma}[theorem]{Lemma}
\newtheorem{proposition}[theorem]{Proposition}
\newtheorem{corollary}[theorem]{Corollary}
\theoremstyle{definition}
\newtheorem{definition}[theorem]{Definition}
\theoremstyle{remark}
\newtheorem{remark}[theorem]{Remark}
\newcommand{\norm}[1]{\left\lVert #1 \right\rVert}

\newcommand{\simplex}[1]{\Delta^{#1}}
\DeclareMathOperator*{\relint}{relint}
\newcommand{\KL}{D_{\mathrm{KL}}}
\newcommand{\TV}[2]{\lno{ #1 - #2 }_{\textrm{TV}} }
\newcommand{\tsphere}[1]{\mathrm{B}^{#1}}
\newcommand{\mycase}[1]{\underline{\textbf{Case #1:}}}
\usepackage{hyperref}
\usepackage[capitalize,noabbrev]{cleveref}
\renewcommand{\thefootnote}{\fnsymbol{footnote}}
\title{Generalized Pinsker Inequality \\for Bregman Divergences of Negative Tsallis Entropies}
\author{
    Guglielmo Beretta\textsuperscript{\rm (1),(2)},
    Tommaso Cesari\textsuperscript{\rm (3)},
    Roberto Colomboni\textsuperscript{\rm (4),(5)}
    \\\\
\textnormal{\textsuperscript{\rm (1)}DAIS, Universit\`{a} Ca' Foscari Venezia, Italy}\\
\textnormal{\textsuperscript{\rm (2)}DAUIN, Politecnico di Torino, Italy}\\
\textnormal{\textsuperscript{\rm (3)}School of Electrical Engineering and Computer Science, University of Ottawa, Canada}\\
\textnormal{\textsuperscript{\rm (4)}DEIB, Politecnico di Milano, Italy}\\
\textnormal{\textsuperscript{\rm (5)}Department of Computer Science, Università degli Studi di Milano, Italy}\\\\
\textnormal{\texttt{\{guglielmo.beretta\}@unive.it}}\\
\textnormal{\texttt{\{tommaso.cesari\}@uottawa.ca}}\\
\textnormal{\texttt{\{roberto.colomboni\}@polimi.it}}\\
}
\begin{document}
\maketitle
\begin{abstract}
The Pinsker inequality lower bounds the Kullback--Leibler divergence $\KL$ in terms of total variation and provides a canonical way to convert $\KL$ control into $\lno{\cdot}_1$-control.
Motivated by applications to probabilistic prediction with Tsallis losses and online learning, we establish a generalized Pinsker inequality for the Bregman divergences $D_\alpha$ generated by the negative $\alpha$-Tsallis entropies---also known as $\beta$-divergences.
Specifically, for any $p$, $q$ in the relative interior of the probability simplex $\simplex{K}$, we prove the sharp bound
\[
    D_\alpha(p\Vert q) \ge \frac{C_{\alpha,K}}{2}\cdot \|p-q\|_1^2 \;,
\]
and we determine the optimal constant $C_{\alpha,K}$ explicitly for every choice of $(\alpha,K)$.
\end{abstract}
\renewcommand*{\thefootnote}{\arabic{footnote}}
\setcounter{footnote}{0}
\section{Introduction and Related Works}
Proper losses\footnote{Or equivalently, proper scoring rules, up to sign.} are a canonical framework for \emph{probabilistic prediction}, where 
given a set of possible outcomes $[K] \coloneqq \{1,\dots,K\}$ for some integer $K \ge 2$, the learner predicts a probability distribution $q$ over $[K]$ and then incurs a loss $\ell(q,y)$ when an outcome $y\in[K]$ is realized.
 The loss $\ell$ is called \emph{strictly proper} if, for every (true) distribution $p$, the expected loss $q \mapsto \E_{Y\sim p} \bsb{ \ell(q,Y) }$ is uniquely minimized at $q=p$ \citep{gneiting-raftery2007jasa}, and  
if the loss is also differentiable, then 
 the \emph{excess risk}
$
    \E_{Y\sim p} \bsb{ \ell(q,Y) }-\E_{Y\sim p}\bsb{ \ell(p,Y) }
$
admits a classical divergence representation: it equals the Bregman divergence generated by the convex function $-H$, where $H$ is the associated Bayes risk (or generalized entropy)
$
    H \colon p \mapsto  \E_{Y\sim p} \bsb{ \ell(p,Y)}
$
\citep{reid2009surrogate,abernethy12}.

Since learning guarantees often control excess risk, a common next step is to translate bounds on the excess risk into more interpretable notions of distributional closeness, such as total variation distance \citep{reid2011information,haghtalab2019losses}.
A canonical example is the log loss
$
    \ell_1(q,k)
\coloneqq
    -\ln q_k
$.
Its Bayes risk is the Shannon entropy $S_1$, and its excess risk coincides with the Kullback--Leibler (KL) divergence $\KL$ on the probability simplex $\simplex{K}$.
The sharp conversion from excess risk to total variation is given by the classic Pinsker inequality from information theory \citep{Tsybakov2008}
\[
    \TV{p}{q}
\le
    \sqrt{\tfrac{1}{2}\cdot\KL(p\Vert q)}\;,
\]
where $\lno{\cdot}_{\mathrm{TV}}$ is the total variation,%
    \footnote{As stated by \cite{reid2011information}, one should be careful in comparing notations from different authors. For instance, the term ``variational divergence'' by \cite{reid2011information} denotes the norm $\lno{\cdot}_1$.}
or equivalently (see \Cref{s:TVvsEll1}),
\[
    \KL(p\Vert q)
\ge
    \tfrac{1}{2}\cdot \lno{p-q}_1^2\;.
\]
In many settings, one departs from the log loss and Shannon entropy, and other strictly proper
losses become useful \citep{Savage1971,LambertPennockShoham2008,giummole2019objective}.
Prominent examples are the \emph{Tsallis} (or \emph{power}) losses
\begin{equation}
\label{eq:tsallis_losses}    
    \ell_\alpha(q,k)
\coloneqq
    \frac{q^{\alpha-1}_k}{1-\alpha}+\frac{1}{\alpha}\sum_{i \in [K]}q_i^\alpha\;,
\end{equation}
(with the ``limiting'' $\alpha = 0$ case, 
$
    \ell_0(q,k)
\coloneqq
    q_k^{-1} -K +\sum_{i \in [K]} \ln q_i
$),
which have been extensively studied in the scoring-rule literature, 
with a variety of immaterial distinct additive or multiplicative normalizations \citep{dawid2007geometry,cichocki_FamiAlphBeta_2010,dawid2014theory,dawid2016minimum,BasuEtal1998DensityPower,kanamori2014,Mameli2015}.
 In the applied literature, these losses appear primarily via the \emph{induced discrepancy}, i.e., the expected excess risk
$\E_{Y\sim p}[\ell_\alpha(q,Y)]-\E_{Y\sim p}[\ell_\alpha(p,Y)]$,
which matches (up to reparametrization and constants independent of $q$) objectives known as \emph{density power divergence} in robust statistics
\citep{BasuEtal1998DensityPower} and as \emph{$\beta$-divergence} in signal processing and matrix factorization
\citep{fevotte2011algorithms}; see, e.g., \cite{dawid2016minimum} and \cite{jewson2018generalBayes}
for an explicit bridge between the Tsallis score and $\beta$-/density power divergences.
Concrete applications include robust inference for health-science data via minimum density power divergence
\citep{nandy2022healthDPD}, hyperspectral unmixing in remote sensing \citep{fevotteDobigeon2015Hyperspectral},
dynamic PET factor analysis \citep{cavalcanti2019PETbeta}, and audio source separation leveraging nonnegative matrix factorization 
\citep{ozerovFevotte2010AudioNMF}.

The Bayes risk of $\ell_\alpha$ is the
$\alpha$-Tsallis entropy $S_\alpha$,
while the associated excess risk is the Bregman divergence $D_\alpha(p\Vert q)$ of $-S_\alpha$ (see \Cref{app:loss_entropy_bregman}). 
Consequently, any Pinsker-type inequality
for the Bregman divergence of $-S_\alpha$ would convert any excess-risk bound on $\ell_\alpha$
into the more interpretable notion of $\lno{\cdot}_1$ (or total variation) control on predictive distributions, and in turn would provide excess-risk bounds for the $0$--$1$ loss of the corresponding plug-in rules via the standard intermediate inequality
$
    \Pb_{Y\sim p} \lsb{ Y \neq k_q^\star }
    -
    \Pb_{Y\sim p}\lsb{Y \neq k_p^\star}
\le
    \lno{p-q}_1
$,
where for any $q \in \Delta^K$ the index $k_q^\star$ is selected deterministically in $\argmax_{k\in[K]} q_k$ (see \Cref{app:0-1}).

Tsallis entropies also arise as regularizers in online learning/convex optimization and multi-armed bandit problems, where they induce geometries and data-dependent behavior that interpolate the Shannon entropy ($\alpha=1$) and log-barrier/Burg entropy ($\alpha=0$) regularizations, and where the associated Bregman divergences play a crucial role in the analysis of the corresponding algorithms (see, e.g., \citealt{abernethy2015fighting,ZimmertSeldin2019TsallisINF,ZimmertSeldin2021BestOfBothWorlds,cesa2022nonstochastic,masoudian21a,kulis2010implicit}).
In these settings, a Pinsker-type inequality for the Bregman divergence of $-S_\alpha$ is equivalent to its $\lno{\cdot}_1$-strong convexity, thereby identifying the curvature parameter that enters standard Follow-the-Regularized-Leader and Mirror-Descent regret analyses (see, e.g., \citealt[Theorem~2.11]{shalev2012online}).

We are therefore led to this question:

\begin{quote}\centering\itshape
What is the analogue of the Pinsker inequality\\ for the
Bregman divergences generated by negative $\alpha$-Tsallis entropies?
\end{quote}
We answer this question by establishing a sharp generalization of Pinsker inequality (\Cref{thm:pinsker}).
Denoting by $D_\alpha(\cdot\Vert\cdot)$ the Bregman divergence generated by the negative $\alpha$-Tsallis entropy $-S_\alpha$ (equivalently, the $\beta$-divergence with $\beta = \alpha$), we provide an explicit \emph{sharp} constant $C_{\alpha,K}$ such that, for any $p,q$ in the relative interior%
    \footnote{
    For fixed $p$ in the relative interior of $\simplex{K}$,
    the divergence $D_\alpha(p\Vert q)$ admits a well-defined (possibly infinite) extension
    to any $q$ in $\simplex{K}$ by passing to the limit,
    and so the inequality trivially
    extends to all $q$ in $\simplex{K}$.
    In what follows, for brevity, we state our results only for $p$ and $q$ that both belong to the relative interior of $\simplex{K}$.
    }
 of $\simplex{K}$,
\[
    D_\alpha(p\Vert q)
\ge
    \frac{C_{\alpha,K} }{2} \cdot \lno{p-q}_1^2\;,
\]
\begin{table}[t]
\centering
\caption{Sharp Pinsker inequality constants $C_{\alpha,K}$ in $D_\alpha(p\Vert q)\ge \frac{C_{\alpha,K}}{2} \cdot\lno{p-q}_1^2$ for the Bregman divergence of the negative $\alpha$-Tsallis entropy on the relative interior of $\simplex{K} $.}
\label{tab:CalphaK-regimes}
\begin{tabular}{@{}lll@{}}
\toprule
Regime & Constant $C_{\alpha,K}$ & Notes \\ \midrule
$\alpha \in (-\infty,1]$
& $\displaystyle 2^{1-\alpha}$
& Dimension-free constant \\[6pt]

$\alpha\in(1,2]$, $K$ even
& $\displaystyle K^{1-\alpha}$
& Polynomial dependence on $K$ \\[6pt]

$\alpha\in(1,2]$, $K$ odd
& $\displaystyle K^{1-\alpha}\cdot \sigma_{\alpha, K}$
& $1 \le \sigma_{\alpha, K} \le 1+
        \frac{7(\alpha-1)}{6(3-\alpha)}\frac{1}{K^2} \le 1 + \frac{7}{6} 
        \frac{1}{K^2} $ 
        \\[6pt]

$\alpha>2$, $K=2$
& $\displaystyle 2^{1-\max\lcb{\alpha, 3}} $
& \text{$\frac{1}{4}$ for $2<\alpha\le 3$, with a phase change at $\alpha = 3$} \\[6pt]

$\alpha>2$, $K\ge 3$
& $\displaystyle 0$
& No Pinsker: $\inf_{p\neq q} D_\alpha(p\|q)/\|p-q\|_1^2=0$ \\ \bottomrule
\end{tabular}
\end{table}
To do so, we derive a variational characterization of $C_{\alpha,K}$ in terms of the Hessian $H S_\alpha$ of $S_\alpha$, which reduces the sharp-constant computation to optimizing a parametric quadratic form over tangent $\lno{\cdot}_1$-unit directions.
As a special case, since the Bregman divergence $D_1$ of the negative Shannon entropy $-S_1$ coincides with the Kullback-Leibler divergence $\KL$, our variational framework recovers Pinsker inequality with the optimal constant for $\alpha = 1$.

The constant $C_{\alpha,K}$ exhibits several qualitatively different regimes, reflecting a
change in the geometry of the Bregman divergence of $-S_\alpha$ as $(\alpha,K)$ varies (see \Cref{tab:CalphaK-regimes}, Figure~\ref{fig:CalphaK-regimes}, and the discussion at the beginning of \cref{sec:proof_Pinsker_new}).
For $\alpha \in (-\infty,1]$, the bound is \emph{dimension-free}: $C_{\alpha,K}=2^{1-\alpha}$ does not depend on $K$,
so Tsallis excess risk controls the $\lno{\cdot}_1$ metric uniformly over the simplex with a constant that depends only on $\alpha$.
In contrast, for $\alpha\in(1,2]$, the constant exhibits power-law decay in the dimension, as 
$C_{\alpha,K} \approx K^{1-\alpha}$, so the excess-risk-to-$\lno{\cdot}_1$ conversion degrades with $K$.
Also, in this case, there is a peculiar parity effect: when $K$ is even, we have the exact identity
$C_{\alpha,K}=K^{1-\alpha}$, whereas for odd $K$ the sharp constant is
$C_{\alpha,K}=K^{1-\alpha}\cdot \sigma_{\alpha, K}$ with $1 \le \sigma_{\alpha, K} = 1 + O(1/K^2)$.
Thus, the even and odd cases agree to the leading order, but the odd case incurs a small second-order correction
that vanishes as $K\to\infty$.
Finally, for $\alpha>2$ the picture changes abruptly.
For multiclass problems, i.e., for $K\ge 3$, there is \emph{no} uniform Pinsker-type inequality of the form
$D_\alpha(p\Vert q)\ge c \lno{p-q}_1^2$ with $c>0$, and additional constraints on $p$ and/or $q$ are necessary to produce a similar inequality (see \Cref{app:clipping}).
By contrast, in the binary case, i.e., for $K=2$, which is arguably the most important special case, a Pinsker-type inequality \emph{does} hold for all $\alpha>2$, so the excess-risk-to-$\lno{\cdot}_1$ conversion remains available (e.g., for binary classification) even beyond the threshold $\alpha=2$.

\paragraph{Other Generalized Pinsker Inequalities.}
Pinsker-type inequalities have been studied for Csisz\'ar $f$-divergences and hence for Tsallis \emph{relative entropies} (which are $f$-divergences) \citep{ReidWilliamson2009GeneralisedPinsker,Gilardoni,rioul2023historical}.
We remark that these results are \emph{orthogonal} to ours: for $\alpha\neq 1$, the Bregman divergences generated by negative $\alpha$-Tsallis entropies are \emph{not} Csisz\'ar $f$-divergences, and they \emph{only} coincide with an $f$-divergence in the Shannon limit $\alpha\to 1$, where the induced Bregman divergence reduces to the Kullback-Leibler divergence $\KL$.
Indeed, the Bregman divergences of negative $\alpha$-Tsallis entropies are decomposable, and on the probability simplex the \emph{only} Csisz\'ar $f$-divergences that are also decomposable Bregman divergences are (up to positive multiplicative constants) $\KL$ and its reverse $(p,q)\mapsto\KL(q\Vert p)$ \citep{Amari_KL_Bregman_f_divergence}.

For background on Csisz\'ar $f$-divergences and $\alpha$-Tsallis relative entropies, as well as a discussion of the relationship between the generalized Pinsker inequality for $\alpha$-Tsallis relative entropies and 
our results, see \Cref{sec:f-divergences}.

\section{Pinsker Inequality for Bregman Divergences of Tsallis Entropies}
For the remainder of this work, we assume that $K$ is any integer greater than or equal to 2.
For any $\alpha > 0$ and any $p \in \R^K$, we denote by $\lno{p}_\alpha \coloneqq \brb{\sum_{k \in [K]} \labs{p_k}^\alpha}^{1/\alpha}$ the usual Minkowski's (quasi-) norm of the vector $p$.
We denote by 
$
    \simplex{K} 
\coloneqq 
    \bcb{ p\in [0,1]^K : \lno{p}_1=1 }
$
 the $(K-1)$-dimensional probability simplex
and its relative interior by
$
    \relint\lrb{\simplex{K}}
\coloneqq 
    \simplex{K} \cap(0,1)^K
\;.
$

We begin by recalling the entropies that appear as Bayes risks for the family of power/Tsallis losses; their negatives serve as convex generators of the Bregman divergences we study.
\begin{definition}[Tsallis entropies]\label{def:tsallis}
For any $\alpha \in \R$, define $S_\alpha \colon [0,+\infty)^K \to [-\infty,+\infty)$ by
\begin{equation}
\label{eq:Tsallis}
    S_\alpha(p) 
\coloneqq
    \begin{cases}
    \vspace{1ex}
        \displaystyle\frac{\sum_{k\in [K]}p_k^\alpha}{\alpha(1-\alpha)}  \;,\qquad &\text{ if $\alpha \notin\lcb{0,1}$}\;,\\
    \vspace{1ex}
    \displaystyle\sum_{k \in [K]} \ln p_k \;,\qquad  &\text{ if $\alpha = 0$}\;,\\
    \displaystyle- \sum_{k \in [K]} p_k \ln p_k \;,\qquad  &\text{ if $\alpha = 1$}\;,
\end{cases}    
\end{equation}
with the understanding that $0^{\alpha} = + \infty$ for negative $\alpha$, and $\ln(0) = -\infty$ but $0 \cdot \ln(0) = 0$.
For any $p\in \simplex{K}$, we call the quantity $S_{\alpha}(p)$ the $\alpha$-\emph{Tsallis entropy}%
\footnote{
    In the literature, definitions of the $\alpha$-Tsallis entropies may differ by a multiplicative factor and by the addition of an affine function of $p$ (in particular, an additive constant).
    For Bregman divergences, scaling the generator scales the divergence, while adding any affine function leaves it unchanged: for $c>0$, $a\in\mathbb{R}$, and $u\in\mathbb{R}^K$,
    \[
    D_{c\cdot f + a + \langle u,\cdot\rangle}(p, q)=c \cdot D_f(p, q)\;,\qquad \text{for any $p,q \in (0,+\infty)^K$}\;.
    \]
    Accordingly, changing the multiplicative factor in the definition of $S_\alpha$ results in the same multiplicative scaling in the induced divergence---and hence scales the constant $C_{\alpha,K}$ in \Cref{thm:pinsker} by the same factor.
    }
of $p$.
\end{definition}
In the literature, $S_0$ is also known as the Burg (or log) entropy, while $S_1$ is the Shannon entropy.
The reason to include $S_0$ and $S_1$ in the Tsallis family is a continuity argument that will become apparent in Remark~\ref{r:beta-divergences}. A common property of all these entropies is given in the next remark, which justifies all subsequent differential calculations.
\begin{remark}
For any $\alpha \in \R$, note that $-S_\alpha$ is strictly convex and of class $C^\infty$ once restricted to $(0,+\infty)^K$. The strict convexity follows immediately from the Hessian of $-S_\alpha$ being positive definite (see Remark~\ref{remark:hessian}).
\end{remark}
We are now ready to introduce our main objects of study: the Bregman divergences generated by $-S_\alpha$.
From the scores/losses perspective, they represent the excess risk of Tsallis scores/losses whose Bayes risk is $S_\alpha$ (see Appendix~\ref{app:loss_entropy_bregman}).
\begin{definition}[Bregman divergences]\label{def:bregman}
Let $f \colon (0,+\infty)^K \to \R$ be a differentiable function.
The Bregman divergence of $f$ is the function
\[
    D_f \colon (0,+\infty)^K\times(0,+\infty)^K\to\R\;,
\qquad
    (p,q) \mapsto f(p)-f(q)- \langle \nabla f(q),p-q\rangle\;,
\]
where $\nabla f(q)$ is the gradient of $f$ in $q$. In what follows, for any $\alpha \in \R$, and any $p,q \in (0,+\infty)^K$, we denote $D_{-S_\alpha}(p,q)$ with the less cumbersome notation $D_\alpha(p \Vert q)$.
\end{definition}
Readers who are not familiar with Bregman divergences may find the next remark useful.
\begin{remark}\label{r:breg-strict-convexity}
For every $\alpha \in \R$, the strict convexity of $-S_{\alpha}$ on $(0,+\infty)^K$ implies that $D_\alpha(p\Vert q) \geq 0$ for every $p, q \in (0,+\infty)^K$, and that $D_\alpha(p\Vert q) = 0$ if and only if $q = p$.
\end{remark}
From the losses perspective, Remark~\ref{r:breg-strict-convexity}  can be read  as saying that the Tsallis losses are strictly proper: the associated excess-risk divergence $D_\alpha$ is nonnegative and vanishes only at $q=p$.

The case $\alpha = 1$ has an interesting known relationship with the Kullback-Leibler divergence.
\begin{remark}
    \label{r:kullback-leibler}
    For all $p,q \in \relint\lrb{\simplex{K}}$, $D_1(p\Vert q) = \KL(p\Vert q)$. Indeed:
\[
    D_1(p\Vert q)
=
    \sum_{k \in [K]} \lrb{ p_k \cdot\ln\frac{p_k}{q_k} - p_k + q_k}
=
    \sum_{k \in [K]}  p_k \cdot\ln\frac{p_k}{q_k}
=
    \KL(p\Vert q)
    \;.
\]
\end{remark}
Moreover, Definition~\ref{def:bregman} is related to the following one-parameter family of $\beta$-divergences from the statistics literature \citep{hennequin2010beta,fevotte2011algorithms}.
\begin{definition}[$\beta$-divergences]
\label{def:beta}
For any $\beta \in \R$, the $\beta$-divergence of $p,q \in (0,+\infty)^K$ is defined as
\[
    d_\beta(p \Vert q)
\coloneq
    \begin{cases}
    \vspace{1ex}
        \displaystyle    \frac{1}{\beta} \cdot \sum_{k \in [K]} \frac{p_k^\beta + (\beta-1)q_k^\beta-\beta p_kq_k^{\beta-1}}{\beta-1}
    \;,\qquad &\text{ if $\beta \notin\lcb{0,1}$}\;,\\
    \vspace{1ex}
    \displaystyle\sum_{k=1}^K
    \lrb{ \frac{p_k}{q_k} -\ln\frac{p_k}{q_k} -1 } \;,\qquad  &\text{ if $\beta = 0$}\;,\\
    \displaystyle \sum_{k \in [K]} \lrb{ p_k \cdot\ln\frac{p_k}{q_k} - p_k + q_k} \;,\qquad  &\text{ if $\beta = 1$}\;,
\end{cases}    
\]
\end{definition}
The divergence $d_0$ is also known as the Itakura-Saito divergence, whereas $d_1$ as the I-divergence.
The link between Definition~\ref{def:bregman} and Definition~\ref{def:beta} is explained in the next remark, which was already noted by \cite{hennequin2010beta}.
\begin{remark}
    \label{r:beta-divergences}
For any $\alpha \in \R$, the Bregman divergence $D_\alpha$ of $-S_\alpha$ coincides with the \emph{$\beta$-divergence} $d_\beta$, with $\beta = \alpha$, i.e., for every $p,q \in (0,+\infty)^K$, and $\alpha \in \R$
\[
    D_\alpha(p\Vert q) = d_\alpha(p\Vert q)\;.
\]
Furthermore, for every $p,q \in (0,+\infty)^K$ and every $\alpha \in \R$, it holds that
\[
    D_\alpha(p \Vert q) = \lim_{\alpha' \to \alpha} D_{\alpha'}(p \Vert q)\;.
\]
\end{remark}

We now have all the ingredients to state our main result.
Our goal is to establish a Pinsker-type lower bound of $D_\alpha(p\Vert q)$ in terms of $\|p-q\|_1^2$ on $\relint\lrb{\simplex{K}}$.
The next theorem identifies the sharp constants \(C_{\alpha,K}\) and gives its explicit expression in all regimes.
\begin{theorem}[Pinsker Inequality for Bregman divergences of Tsallis Entropies]
\label{thm:pinsker}
Let $\alpha \in \R$ and denote by $C_{\alpha, K}$ the largest $C \geq 0$ such that for every $p$, $q \in  \relint\lrb{\simplex{K}}$
\begin{equation}    
    D_\alpha\brb{p\Vert q}
\ge
    \dfrac{C}{2} \cdot\lno{p - q}_1^2\;.
\end{equation}    
Then,
\begin{align}
    C_{\alpha,K} = 
    \begin{cases}
        2^{1-\alpha} \qquad &\textit{if $\alpha \in (-\infty,1]$,}
        \\
         K^{1-\alpha} \qquad &\textit{if $\alpha \in (1,2]$ and $K$ is even,}
        \\
        K^{1-\alpha} \cdot \sigma_{\alpha, K}  &\textit{if $\alpha \in (1,2]$ and $K$ is odd,}
        \\
         2^{1-\max\lcb{\alpha, 3}} &\textit{if $\alpha \in (2,+\infty)$ and $K = 2$, }
        \\
        0 &\textit{if $\alpha \in (2,+\infty)$ and $K \ge 3$,} 
    \end{cases}
\end{align}
where $\sigma_{\alpha, K} \coloneq \lrb{  \frac{ (1-\frac{1}{K})^{\frac{1 - \alpha}{3- \alpha}} + (1+\frac{1}{K})^{\frac{1 - \alpha}{3- \alpha}} }{2} }^{3 - \alpha}$ (see Remark~\ref{rem:sigma} for more about this term).
\end{theorem}
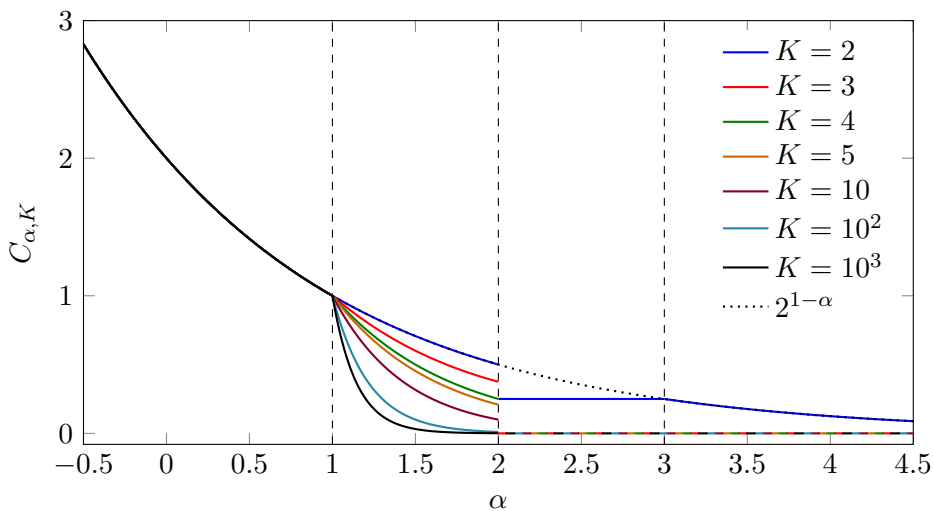
\begin{figure}[t]
\centering
\begin{tikzpicture}
\begin{axis}[
    width=12.5cm,
    height=7.2cm,
    xlabel={$ \alpha $},
    ylabel={$ C_{\alpha,K} $},
    xmin=-0.5, xmax=4.5,
    ymin=-0.08, ymax=3,
    ytick={0,1,2,3},
    title={Sharp Pinsker constants $C_{\alpha,K}$ (\Cref{thm:pinsker})},
    legend cell align=left,
    legend style={draw=none, cells={anchor=west}},
]

\addplot[black, dashed, forget plot] coordinates {(1,-0.08) (1,3)};
\addplot[black, dashed, forget plot] coordinates {(2,-0.08) (2,3)};
\addplot[black, dashed, forget plot] coordinates {(3,-0.08) (3,3)};

\addplot[thick, blue, domain=-0.5:2, samples=400] {pow(2,1-x)};
\addlegendentry{$K=2$}
\addplot[thick, blue, domain=2:3, samples=2, forget plot] {0.25};
\addplot[thick, blue, domain=3:4.5, samples=250, forget plot] {pow(2,1-x)};

\def\K{3}
\addplot[thick, red, domain=-0.5:1, samples=250] {pow(2,1-x)};
\addlegendentry{$K=3$}
\addplot[thick, red, domain=1:2, samples=350, forget plot]
{
    pow(\K,1-x) *
    pow(
        0.5 * (
            pow(1-1/\K, (1-x)/(3-x)) +
            pow(1+1/\K, (1-x)/(3-x))
        ),
        3-x
    )
};

\def\K{4}
\addplot[thick, green!50!black, domain=-0.5:1, samples=250] {pow(2,1-x)};
\addlegendentry{$K=4$}
\addplot[thick, green!50!black, domain=1:2, samples=350, forget plot] {pow(\K,1-x)};

\def\K{5}
\addplot[thick, orange!80!black, domain=-0.5:1, samples=250] {pow(2,1-x)};
\addlegendentry{$K=5$}
\addplot[thick, orange!80!black, domain=1:2, samples=350, forget plot]
{
    pow(\K,1-x) *
    pow(
        0.5 * (
            pow(1-1/\K, (1-x)/(3-x)) +
            pow(1+1/\K, (1-x)/(3-x))
        ),
        3-x
    )
};

\def\K{10}
\addplot[thick, purple!70!black, domain=-0.5:1, samples=250] {pow(2,1-x)};
\addlegendentry{$K=10$}
\addplot[thick, purple!70!black, domain=1:2, samples=350, forget plot] {pow(\K,1-x)};

\def\K{100}
\addplot[thick, cyan!60!black, domain=-0.5:1, samples=250] {pow(2,1-x)};
\addlegendentry{$K=10^2$}
\addplot[thick, cyan!60!black, domain=1:2, samples=350, forget plot] {pow(\K,1-x)};

\def\K{1000}
\addplot[thick, black, domain=-0.5:1, samples=250] {pow(2,1-x)};
\addlegendentry{$K=10^3$}
\addplot[thick, black, domain=1:2, samples=350, forget plot] {pow(\K,1-x)};

\addplot[
  thick, red,
  dash pattern=on 3pt off 15pt,
  dash phase=0pt,
  line cap=round,
  domain=2:4.5, samples=2,
  forget plot
] {0};

\addplot[
  thick, green!50!black,
  dash pattern=on 3pt off 15pt,
  dash phase=3pt,
  line cap=round,
  domain=2:4.5, samples=2,
  forget plot
] {0};

\addplot[
  thick, orange!80!black,
  dash pattern=on 3pt off 15pt,
  dash phase=6pt,
  line cap=round,
  domain=2:4.5, samples=2,
  forget plot
] {0};

\addplot[
  thick, purple!70!black,
  dash pattern=on 3pt off 15pt,
  dash phase=9pt,
  line cap=round,
  domain=2:4.5, samples=2,
  forget plot
] {0};

\addplot[
  thick, cyan!60!black,
  dash pattern=on 3pt off 15pt,
  dash phase=12pt,
  line cap=round,
  domain=2:4.5, samples=2,
  forget plot
] {0};

\addplot[
  thick, black,
  dash pattern=on 3pt off 15pt,
  dash phase=15pt,
  line cap=round,
  domain=2:4.5, samples=2,
  forget plot
] {0};

\addplot[dotted, thick, domain=-0.5:4.5, samples=400] {pow(2,1-x)};
\addlegendentry{$2^{1-\alpha}$}

\end{axis}
\end{tikzpicture}
\caption{\label{fig:CalphaK-regimes}Sharp Pinsker constants $C_{\alpha,K}$ for selected $K$. Dashed lines mark $\alpha=1,2,3$ where the phase transitions occur.}
\end{figure}
Note that, in light of Remark~\ref{r:kullback-leibler} and to the best of our knowledge, we also establish for the first time a sharp Pinsker-type inequality for \emph{all} $\beta$-divergences (including the Itakura-Saito divergence) for $\beta \neq 1$, and we recover the classic (sharp) Pinsker inequality for $\beta = 1$
\[
    \KL (p\Vert q)
=
    D_1 (p\Vert q)
\ge
    \frac{1}{2} \cdot  \norm{p-q}_1^2\;, \qquad \text{for all $p,q\in\relint\lrb{\simplex{K}}$}\;.
\]
As an immediate corollary of \Cref{thm:pinsker}, we also obtain the $\lno{\cdot}_1$-strong convexity constant of $-S_\alpha$, i.e., the curvature parameter $\mu$ that appears in standard Follow-the-Regularized-Leader and mirror-descent analyses with Tsallis regularization.
\begin{corollary}\label{cor:strong-convexity}
The negative Tsallis entropy $-S_\alpha$ is $ C_{\alpha,K}$-strongly convex with respect to $\lno{\cdot}_1$.
Moreover, this value $C_{\alpha,K}$ is optimal, i.e., it is the largest strong convexity parameter for which the inequality holds.
\end{corollary}
We now present a series of remarks commenting on the phase transition behavior for $C_{\alpha,K}$.
Beyond their intrinsic information-theoretic interest, these regimes clarify how \emph{dimension}, \emph{the choice of $\alpha$}, and \emph{the binary vs.\ multiclass setting} determine the strength of the conversion from Tsallis-score excess risk (equivalently, the divergence $D_\alpha$) into $\lno{\cdot}_1$ control.
\begin{remark}
For $\alpha \le 1$ the value of $C_{\alpha, K}$ does not depend on the dimension $K$.
\end{remark}
This means that, for $\alpha \le 1$, the excess-risk-to-$\lno{\cdot}_1$ conversion in \Cref{thm:pinsker} is dimension-free, implying, in particular, that in large multiclass settings, $\alpha \le 1$ avoids a curse of dimensionality in this conversion.

The next remark quantifies the parity correction $\sigma_{\alpha,K}$ for $\alpha\in(1,2]$ and odd~$K$, and makes precise that $C_{\alpha,K}=K^{1-\alpha}$ up to the lower-order factor $\sigma_{\alpha,K}=1+O(K^{-2})$.
\begin{remark}
\label{rem:sigma}
For $\alpha \in (1,2]$ and $K$ odd, the constant $C_{\alpha, K}$ equals $K^{1-\alpha}$ multiplied by the scaling term $\sigma_{\alpha, K}$, which satisfies $\sigma_{\alpha, K} > 1$ and $\sigma_{\alpha, K} \to 1$ as $\alpha \to 1^+$. We also have the estimates
\[
    1 \le 1 + \frac{\alpha - 1}{3- \alpha} 
        \cdot\frac{1}{K^2} 
\le
    \sigma_{\alpha, K}
\le
    1 + \frac{7(\alpha - 1)}{6(3- \alpha)} 
        \cdot\frac{1}{K^2}
\le
    1 + \frac{7}{6} \cdot
        \frac{1}{K^2}\;.
\]
See Appendix~\ref{app:sigma} for a proof of this. In particular, $K^{1 - \alpha}< C_{\alpha, K} \approx K^{1 - \alpha}$ for $\alpha \in (1,2]$ and $K$ odd.
\end{remark}
Another curious fact can be noted in the case $\alpha = 2$, which is of particular interest since $D_2$ coincides with half the squared Euclidean distance.\footnote{
Indeed, according to Definition~\ref{def:tsallis}, we have
$
    -S_2(p)
=
    \frac{1}{2}\lno{p}_2^2,
$
and hence, for any $p,q \in \relint\lrb{\simplex{K}}$, we have 
$
    D_2(p\Vert q)
=
    \frac{1}{2}\lno{p}_2^2-\frac{1}{2}\lno{q}_2^2 - \lan{q,p-q}
=
    \frac{1}{2}\lno{p}_2^2 +\frac{1}{2}\lno{q}_2^2 - \lan{p,q}
=
    \frac{1}{2}\lno{p-q}_2^2\;.
$
}
\begin{remark}
\label{r:euclidean}
Setting $\alpha = 2$ in \Cref{thm:pinsker} yields
\[
    C_{2, K}
=    
        \begin{cases}
            \frac{1}{K}
        &
            \text{for $K$ even}\;,
        \\
            \frac{1}{K} \cdot
    \frac{K^2}{K^2-1}
        &
            \text{for $K$ odd}\;.
        \end{cases}
\]  
Recalling also the relationship between $D_2(p\Vert q)$ and the Euclidean distance
 between $p$ and $q$,
we recover the well-known inequality $\lno{p-q}_2^2 \ge \frac{1}{K} \cdot \lno{p-q}_1^2$, but with the additional observation that the constant $\frac{1}{K}$ is not sharp when $K$ is odd and $p, q$ are confined in $\simplex{K}$.
\end{remark}
For $\alpha>2$, the picture changes qualitatively, and a sharp distinction emerges between the binary case $K=2$ and the multiclass case $K\ge 3$.
In the binary setting, a nontrivial Pinsker-type inequality continues to hold, so small Tsallis-losses excess risk still enforces $\lno{\cdot}_1$-closeness of predictive distributions.
In particular, we note the following:
\begin{remark}
For $K = 2$ the map $\alpha \mapsto C_{\alpha, K}$ is a piecewise smooth function of $\alpha$ on $[0, +\infty)$, whose sole discontinuity is a jump for $\alpha = 2$.
\end{remark}
In contrast, in the multiclass setting $K\ge 3$, the $\alpha>2$ regime is inherently pathological to convert excess-risk bounds into $\lno{\cdot}_1$ ones.
\begin{remark}\label{rem:no-pinsker-K-ge-3}
    Theorem~\ref{thm:pinsker} states that for $\alpha >2$ and $K \geq 3$ it is not possible to find a \emph{positive} constant $C$ such that $D_\alpha(p\Vert q) \ge C\cdot
    \lno{ p-q }_1^2$ for every $p$, $q \in \relint{\lrb{\simplex{K}}}$.
\end{remark}
See also Appendix~\ref{app:nopinsker} for an explicit construction of one-parameter families $p(t),q(t)\in \relint(\simplex{K})$ for which $D_\alpha\brb{p(t)\Vert q(t)}=o\brb{\lno{p(t)-q(t)}_1^2}$, as $t \to 0$.
Despite Remark~\ref{rem:no-pinsker-K-ge-3}, it is possible to prove that a Pinsker-type inequality can still hold when all coordinates of $p$ and $q$ are bounded away from zero. The interested reader may find a proof of this in \Cref{app:clipping}.

We conclude this section by remarking what happens if a Pinsker-type inequality is sought not only over $\relint{\lrb{\simplex{K}}}$, but on the whole $(0,+\infty)^K$ instead.
\begin{remark}[No uniform Pinsker inequality on $(0,+\infty)^K$ except for $\alpha=2$]
\label{rem:nopinsker-orthant}
If $\alpha\neq 2$, there is no constant $C>0$ such that
$D_\alpha(p\Vert q)\ge C\cdot \lno{p-q}_1^2$, for all $p,q\in(0,+\infty)^K$. If $\alpha = 2$, then for all $p,q\in(0,+\infty)^K$ it holds that $D_2(p \Vert q) \ge \frac{1}{2K} \cdot \lno{p-q}_1^2$ and the constant $\frac{1}{2K}$ is optimal.

A proof of this remark can be found in \Cref{app:remarks}.
\end{remark}
\section{Geometric Intuition Behind the Proof of \texorpdfstring{\Cref{thm:pinsker}}{Theorem \ref{thm:pinsker}}}
\label{sec:proof_Pinsker_new}

We now give some insight into the technique we use to derive the explicit expression for the sharp constant $C_{\alpha,K}$.
First, we observe that $C_{\alpha,K}$ admits a variational formulation involving the Hessian
of $S_{\alpha}$ and the independent variables $(\gamma, v)$, where $\gamma$ lies in $\relint{\lrb{\simplex{K}}}$, whereas $v$ is a $\lno{\cdot}_1$-unit vector (i.e., $\lno{v}_1 = 1$) and it is tangent to $\relint{\lrb{\simplex{K}}}$ (i.e., $\sum_{k=1}^K v_k = 0$).
Indeed, $C_{\alpha,K}$ can be written as the infimum over $(\gamma, v)$ of a parametric quadratic form in the variable $v$ whose coefficients (i.e.,  $\gamma_k^{\alpha-2}$) depend on $\gamma$ and $\alpha$.
Of main importance is the qualitative behavior of these coefficients, which exhibits a phase transition at $\alpha = 2$.
Specifically, in the case $\alpha \le 2$ the coefficients admit a positive lower bound (since $\gamma_k < 1$ and $\alpha - 2 \le 0$), but in the case $\alpha > 2$ this is no longer true (since $\alpha-2>0$ and the corresponding $\gamma_k$ approaches $0$).

To handle the case $\alpha \le 2$, we first fix $v$ in the parametric quadratic form and take the infimum over all possible values of $\gamma$, which amounts (up to a multiplicative constant) to finding the minimal second directional derivative of $-S_{\alpha}$ in the direction $v$ over the relative interior of the simplex $\relint\lrb{\simplex{K}}$.
It turns out that this infimum equals the square of $\lno{v}_{\beta}$, with $\beta$ being a positive number (specifically, $\beta = \tfrac{2}{3-\alpha}$) that depends on $\alpha$,
and so $C_{\alpha, K}$ is related to minimizing $\lno{\cdot}_{\beta}$ over the tangent $\lno{\cdot}_1$-unit vectors. 
Interestingly, the expression for $C_{\alpha, K}$ reflects the localization of $v$ attaining such constrained minimum. 
Indeed, the case $\alpha \in (-\infty,1)$ yields $0 < \beta < 1$, and $\lno{v}_{\beta}$ is minimized by concentrating the mass of $v$ over two coordinates, which causes $C_{\alpha, K}$ to be independent of $K$.
If $\alpha = 1$, then also $\beta = 1$, and tautologically $\lno{v}_{\beta}= 1$, which results in $C_{1, K} = 1$. This recovers the classical Pinsker inequality.
The case $\alpha \in (1,2]$ yields $\beta > 1$, and $\lno{v}_{\beta}$ is minimized by spreading the mass as evenly as possible across coordinates, producing an expression for $C_{\alpha, K}$ that is decreasing in $K$ and with a subtle role played by the parity of $K$.

For $\alpha > 2$, we shall see that a substantial difference occurs depending on $K$.
In fact, if $K\ge 3$, then it is possible to show that $C_{\alpha, K} = 0$ thanks to the existence of some tangent vector $v$ that has some zero components.
For $K = 2$, a similar $v$ does not exist, and indeed $C_{\alpha, 2}> 0$. This time, a phase transition occurs at $\alpha = 3$, where the map $\gamma_k \mapsto\gamma_k^{\alpha-2}$ transitions from concave to convex.
\section{Toward the Proof: Useful Results}
For any finite set $A$, we use the notation $|A|$ to denote the number of elements in $A$.
\begin{remark}
\label{remark:hessian}
For any $\alpha \in \R$ and $p\in(0,+\infty)^K$, we have that
\[
    \partial^2_{ij} S_\alpha(p)
=
    - p_i^{\alpha-2} \cdot \I\{i = j\}\;,
\]
and hence, the Hessian of $S_\alpha(p)$ is the $K \times K$ matrix $HS_\alpha(p) = -\mathrm{diag}\brb{ p_1^{\alpha-2},\ldots, p_K^{\alpha-2}}$.
A proof of this remark can be found in \cref{app:remarks}.
\end{remark}
We now define the key quantities $C_{\alpha,K}$ appearing in the statement of \Cref{thm:pinsker}.
\begin{definition}\label{def:const}
    For every $\alpha \in \R$, we define
    \begin{equation}\label{eq:sharp}    
        C_{\alpha, K}
    \coloneq
        2 \cdot  \inf \lcb{\frac{D_\alpha\brb{p\Vert q}}{\lno{p - q}_1^2} : p , q \in  \relint\lrb{\simplex{K}}, p \neq q}\;.
    \end{equation}
    \end{definition}
\begin{remark}\label{remark:const-is-best}
    By construction, the constant $C_{\alpha, K}$ given in Definition~\ref{def:const} is the largest $C \geq 0$ such that $D_\alpha\brb{p\Vert q} \ge \frac{C}{2} \cdot \lno{p - q}_1^2$ for every $p$, $q \in  \relint\lrb{\simplex{K}}$.
    A proof of this fact can be found in \cref{app:remarks}.
    The proof of Theorem~\ref{thm:pinsker} amounts to computing an explicit expression for $C_{\alpha, K}$.
\end{remark}
\begin{definition}
    We set
\begin{equation}\label{eq:1-sphere-section}
    \tsphere{K}
\coloneq
    \lcb{v \in \R^K :\lno{v}_1=1, \sum_{k\in[K]}v_k = 0}\;.
\end{equation}
\end{definition}
\begin{remark}    
    For every $\gamma \in \relint{\lrb{\simplex{K}}}$, the vector space $\{ v \in \R^K : \sum_{k \in [K]} v_k = 0\}$ can be thought as the \emph{tangent space} of $\relint{\lrb{\simplex{K}}}$ in $\gamma$, and $\tsphere{K}$ is precisely the unit sphere thereof in the norm $\lno{. }_1$.
\end{remark} 

\begin{proposition}[Reformulation via quadratic forms]\label{prop:second_order}
For every $\alpha \in \R$ we have
\begin{equation}\label{eq:cambio}
     C_{\alpha, K}
=
    \inf \lcb{\sum_{k \in [K]} v_k^2 \gamma_k^{\alpha -2} : v \in \tsphere{K}, \gamma \in \relint\lrb{\simplex{K}}}\;.    
\end{equation}
\end{proposition}
\begin{proof}
Let $C'_{\alpha, K}$ be the RHS of \eqref{eq:cambio}.
Given two distinct $p$, $q \in \relint\lrb{\simplex{K}}$, we consider the following first-order Taylor expansion of $S_{\alpha}(p)$ with Lagrange remainder:
\begin{equation}\label{eq:taylor}
    S_{\alpha}(p)
    = S_{\alpha}\brb{q} + \langle\nabla S_{\alpha}(q), p-q \rangle
    + \frac{1}{2} (p-q)^\top H S_{\alpha}\brb{\gamma} (p-q) \;,
\end{equation}    
with $\gamma \coloneq xp + (1-x) q$ for a suitable $x \in (0,1)$ and with $H S_{\alpha}$ being the Hessian of $S_\alpha$. In particular, $\gamma$ is a convex combination of $p$ and $q$, hence $\gamma \in \relint\lrb{\simplex{K}}$.
Define $v \coloneq \frac{p - q}{\lno{p - q}_1}$, note that $v \in \tsphere{K}$,
and by \eqref{eq:taylor}, by Definition~\ref{def:bregman}, and by Remark~\ref{remark:hessian} we observe that
\begin{equation}\label{eq:change-of-variables}
    2\cdot \frac{D_\alpha\brb{p\Vert q}}{\lno{p - q}_1^2}
=
    -  \lrb{\frac{p - q}{\lno{p - q}_1}}^\top H S_{\alpha}(\gamma) \frac{p - q}{\lno{p - q}_1}
=
    \sum_{k \in [K]} v_k^2 \gamma_k^{\alpha -2}\;,
\end{equation}
i.e., that $2 \cdot D_\alpha\brb{p\Vert q}\lno{p - q}_1^{-2} \ge C'_{\alpha, K}$. Since $p$ and $q$ were arbitrarily chosen, this gives $C_{\alpha, K} \ge C'_{\alpha, K}$. To prove the opposite inequality,  let $\zeta \in \relint\lrb{\simplex{K}}$ and let $u \in \tsphere{K}$. Define for $t>0$ the points $\tilde{p}(t) \coloneq \zeta + t \frac{1}{2} u$ and $\tilde{q}(t) \coloneq \zeta - t \frac{1}{2} u$,
which lie in $\relint\lrb{\simplex{K}}$ if $t$ is small enough 
and which satisfy $u = \frac{\tilde{p}(t) - \tilde{q}(t)}{\lno{\tilde{p}(t) - \tilde{q}(t)}_1}$.
Then, arguing as before, we deduce the existence of some
$r(t)$ that is a convex combination of $\tilde{p}(t)$ and $\tilde{q}(t)$ and such that 
\eqref{eq:change-of-variables} is valid for
$(p, q, v, \gamma) = (\tilde{p}(t), \tilde{q}(t), u, r(t))$.
In particular, $C_{\alpha, K} \le \sum_{k \in [K]} u_k^2 \brb{r_k(t)}^{\alpha -2}$, thus 
$C_{\alpha, K} \le \sum_{k \in [K]} u_k^2 \zeta_k^{\alpha -2}$,
since $r(t) \to \zeta$ as $t \to 0^+$. By the arbitrariness of $\zeta$ and $u$, it follows that $C_{\alpha, K} \le C'_{\alpha, K}$, and so $C_{\alpha, K} = C'_{\alpha, K}$.
\end{proof}
In the sequel, we will need two technical lemmas. The first lemma addresses the minimization of a certain class of homogeneous functionals over $\relint{\lrb{\simplex{K}}}$. The second lemma quantifies how close is $\tsphere{K}$ to the origin with respect to all Minkowski's (pseudo-) norms.
\begin{lemma}[Infimum over the relative interior of the simplex]\label{lemma:mini_sum_convex}
Let $\nu \ge 0$ and consider $\lambda \in \R^K$ such that $\lambda \geq 0$.
Then 
\[
    \inf \lcb{
    \sum_{k \in [K]} \lambda_k\gamma_k^{-\nu} :  \gamma \in \relint{\lrb{\simplex{K}}}
    }
=
    \lrb{\sum_{k \in [K]} \lambda_k^{\frac{1}{\nu + 1}}}^{\nu + 1}\;.
\]
\end{lemma}
\begin{proof}
Define $S \coloneq \{k \in [K]: \lambda_k>0\}$. We assume $S \neq \emptyset$ and $\nu > 0$, for otherwise the proof is trivial.
Define on $(0,+\infty)^{K}$ the function $f(\gamma) \coloneq \sum_{k \in [K]}  \lambda_k \gamma_k^{-\nu}$.
For the case $S = [K]$, note that $f$ is strictly convex on $\relint\lrb{\simplex{K}}$ and $f(\gamma) \to +\infty$ as 
$\gamma_k \to 0$ for any $k$,
thus the restriction of $f$ to $\relint\lrb{\simplex{K}}$ attains its minimum on some $\hat{\gamma} \in \relint\lrb{\simplex{K}}$. By the first-order optimality conditions, there exists some $\mu \in \R$ such that $\partial_k f (\hat{\gamma}) =  - \nu \lambda_k \hat{\gamma}_k^{- \nu - 1} = \mu$ for every $k \in [K]$.  
But then the value of 
$\lambda_k^{- \frac{1}{\nu + 1}} \hat{\gamma}_k$ is independent of $k$, and so $\hat{\gamma}_k = \lrb{\sum_{i \in [K]} \lambda_i^{\frac{1}{\nu + 1}} }^{-1}\lambda_k^{\frac{1}{\nu + 1}}$, and
\begin{align*}
    f(\hat{\gamma})
&=
    \sum_{k \in [K]} \lambda_k
    \lrb{\sum_{i \in [K]} \lambda_i^{\frac{1}{\nu + 1}}}^{\nu}
    \lambda_k^{-\frac{\nu }{\nu + 1}}
=
    \sum_{k \in [K]} \lambda_k^{\frac{1}{\nu + 1}}
    \lrb{\sum_{i \in [K]} \lambda_i^{\frac{1}{\nu + 1}}}^{\nu}
=
     \lrb{\sum_{k \in [K]} \lambda_k^{\frac{1}{\nu + 1}} }^{\nu + 1}\;.
\end{align*}
For the case $S \neq [K]$, write $S = \{ k_1, \dots k_{|S|}\}$
and define the mapping $\pi \colon \relint\lrb{\simplex{K}} \to \R^{|S|}$ given by $\pi_i(\gamma) \coloneq  \gamma_{k_i}/ (\sum_{k \in S} \gamma_k)$, so that
\begin{align}\label{eq:gammatilde}
    f(\gamma)
&=
    \sum_{k \in S}  \lambda_k \gamma_k^{-\nu}
=
    \lrb{\sum_{k \in S} \gamma_k}^{-\nu} \sum_{i = 1}^{|S|}  \lambda_{k_i} \lrb{\pi_i(\gamma)}^{-\nu}\;.
\end{align}
For every $W  \in (0,1)$ define $E_W \coloneq \lcb{\gamma \in \relint\lrb{\simplex{K}}: \sum_{k \in S} \gamma_k= W }$ and note that $\pi$ maps $E_W$ onto $\relint\lrb{\simplex{|S|}}$, hence by \eqref{eq:gammatilde} we get
\begin{align*}
    \inf_{\gamma \in E_W} f(\gamma) 
&=
    W^{-\nu} \inf_{\tilde{\gamma} \in \relint\lrb{\simplex{|S|}}} \sum_{i = 1}^{|S|}  \lambda_{k_i} \tilde{\gamma}_i^{-\nu}
=
    W^{-\nu} \lrb{\sum_{i = 1}^{|S|} \lambda_{k_i}^{\frac{1}{\nu + 1}} }^{\nu + 1}\,,
\end{align*}
where the last equality follows by the case $S = [K]$. In particular, this means that 
\begin{align*}
    \inf_{\gamma \in \relint\lrb{\simplex{K}}} f(\gamma)
&=
    \inf_{W \in (0,1)}
    \inf_{\gamma \in E_W} f(\gamma) 
=
    \inf_{W \in (0,1)}
    W^{-\nu}
    \lrb{\sum_{i = 1}^{|S|} \lambda_{k_i}^{\frac{1}{\nu + 1} }}^{\nu + 1} 
=
     \lrb{\sum_{k \in [K]} \lambda_k^{\frac{1}{\nu + 1}} }^{\nu + 1}\;,
\end{align*}
where we used that 
$\inf_{W \in (0,1)} W^{-\nu}
=
    \lim_{W \to 1^-} W^{-\nu}
=
    1$.
\end{proof}
\begin{lemma}[Changing norm on the tangent space]\label{lemma:nuuuu}
\begin{enumerate}
    \item\label{i:nu-0-1}
        For $0 < \beta <1$ we have
        $
            \min_{v \in \tsphere{K}} \lno{v}_{\beta}
        $
        $
            =  2^{\frac{1}{\beta} - 1}
        $,
        and the minimum is attained at
        \begin{equation}\label{eq:disequalized}
            v_k
        =
            \begin{cases}
                (-1)^k \cdot \frac{1}{2}
                &\text{for }k = 1, 2
        \\
                0
                &\text{otherwise;}
            \end{cases}
        \end{equation}
    \item\label{i:nu-1-infty}
        For $\beta > 1$ we have
        $
            \min_{v \in \tsphere{K}} \lno{v}_{\beta}
        =  
             2^{-1} \lrb{{\left\lfloor \fracc{K}{2} \right\rfloor}^{1 - \beta} +  {\left\lceil \fracc{K}{2} \right\rceil}^{1 - \beta}}^{1/\beta}
       $,
        and the minimum is attained at
        \begin{equation}\label{eq:equalized}
            v_k
        =
            \begin{cases}
                \frac{1}{2{\left\lfloor \frac{K}{2} \right\rfloor}}
                &\text{for }1 \le k \le  \left\lfloor \frac{K}{2} \right\rfloor
        \\
                -\frac{1}{2{\left\lceil \frac{K}{2} \right\rceil}}
                &\text{for }  \left\lfloor \frac{K}{2} \right\rfloor < k \le K.
            \end{cases}
        \end{equation}
\end{enumerate}
\end{lemma}
\begin{proof}
Let $v \in \tsphere{K}$, and set $P \coloneq \{k \in [K]: v_k >0\}$ and $N \coloneq \{k \in [K]: v_k < 0\}$. By definition of $\tsphere{K}$, we have $\sum_{k \in P} |v_k| + \sum_{k \in N} |v_k| = 1$ and
$\sum_{k \in P} |v_k| - \sum_{k \in N} |v_k| = 0$, therefore $\sum_{k \in P} |v_k| = \sum_{k \in N} |v_k| = \frac{1}{2}$. 

\ref{i:nu-0-1}: For $0 < \beta <  1$ the function $x \mapsto x^{\beta}$ is strictly concave on $[0,+\infty)$.
Furthermore, we have $x^\beta + y^{\beta} \geq (x + y)^\beta$ for every $x$, $y \geq 0$, thus
\[
    \sum_{k \in [K] } |v_k|^{\beta}
=
    \sum_{k \in P } |v_k|^{\beta}
+
    \sum_{k \in N } |v_k|^{\beta}
\ge
    \lrb{\sum_{k \in P } |v_k|}^{\beta}
+
    \lrb{\sum_{k \in N } |v_k|}^{\beta}
= 2 \,\cdot \frac{1}{2^{\beta}} = 2^{1 - \beta}\;,    
\]
i.e., $\lno{v}_{\beta} \ge 2^{\frac{1}{\beta} - 1}$.
Finally, note that this inequality is an equality if and only if $|P|=|N|=1$, which holds true for $v$ defined as in \eqref{eq:disequalized}.

\ref{i:nu-1-infty}: For $\beta > 1$ the  function $x \mapsto x^{\beta}$ is strictly convex on $[0,+\infty)$ and Jensen's inequality yields 
\begin{equation}\label{eq:distribute}
    \sum_{k \in P } |v_k|^{\beta}
=
    |P| \cdot \sum_{k \in P } \frac{{|v_k|}^{\beta}}{|P|}
\ge 
    |P| \cdot \lrb{ \sum_{k \in P } \frac{|v_k|}{|P|}}^{\beta}
=
    |P| \cdot \lrb{\frac{1}{2|P|}}^{\beta}
=
    2^{-\beta}\cdot  |P|^{1 - \beta}\;,
\end{equation}
where the equality is attained if and only $v_k$ takes on the same value for every $k \in P$, i.e., if and only if $v_k = 1/\brb{2|P|}$ for every $k \in P$.
Similarly,
    $\sum_{k \in N } |v_k|^{\beta}
    \ge
    2^{-\beta} |N|^{1 - \beta} $,
where the equality is attained if and only if $v_k = -1/\brb{2|N|}$ for every $k \in N$.
We deduce that
\begin{equation}\label{eq:bounds-with-integers}    
    \sum_{k \in [K]} |v_k|^{\beta}
\ge
    2^{-\beta}\cdot \lrb{ |P|^{1 - \beta} +  |N|^{1 - \beta}}\;.
\end{equation}
Now we find what are the feasible values for the integers $|P|$ and $|N|$ that minimize the RHS of \eqref{eq:bounds-with-integers}. Without loss of generality, we assume $|P| \leq |N|$ and we observe that the minimizing values for $|P|$, $|N|$ satisfy the constraint $|P| + |N| = K$, i.e., that $\{ P, N\}$ is a partition of $[K]$, as a consequence of $|N|^{1 -   \beta} \geq (K - |P|)^{1 - \beta}$. 
Next, we observe that the function $\phi(x) \coloneq x^{1 -\beta} + (K - x)^{1- \beta}$ restricted to the interval $(0, K/2]$ is strictly decreasing, and so the minimum of the RHS of \eqref{eq:bounds-with-integers} is obtained for $\brb{|P|, |N|} = \lrb{\lfloor \frac{K}{2}\rfloor,\left\lceil \frac{K}{2} \right\rceil}$ and equals ${ 2^{-\beta}} \cdot \lrb{{\left\lfloor \frac{K}{2}\right\rfloor}^{1 - \beta} +  {\left\lceil \frac{K}{2} \right\rceil}^{1 - \beta}}$. This is precisely the value of $\lno{v}_{\beta}^{\beta}$ for $v$ defined as in \eqref{eq:equalized}.  
\end{proof}
\section{The Proof}
Now we have all the necessary tools to complete the proof of our main theorem.
{\renewcommand{\proofname}{Proof of Theorem~\ref{thm:pinsker}}
\begin{proof}
By Proposition~\ref{prop:second_order}, we know that 
\[
    C_{\alpha, K}
=
    \inf \lcb{ \sum_{k \in [K]} v_k^2 \gamma_k^{\alpha -2} : v \in \tsphere{K}, \gamma \in \relint\lrb{\simplex{K}}}\;.
\]  
Computing this infimum requires different strategies for different values of $\alpha$ and $K$.
\\
\mycase{$\alpha \le 2$}
Observe that $2 -\alpha  \ge 0$, thus by Lemma \ref{lemma:mini_sum_convex}
\[
    C_{\alpha, K}
=
    \inf_{v \in \tsphere{K}}
    \inf_{\gamma \in \relint\lrb{\simplex{K}}}
    \sum_{k \in [K]} v_k^2 \gamma_k^{\alpha -2}
=
    \inf_{v \in \tsphere{K}} \lrb{\sum_{k \in [K]} |v_k|^{\frac{2}{3 -\alpha}} }^{3-\alpha}
=
    \inf_{v \in \tsphere{K}} \lno{v}_{\beta}^2 \;,
\]
where we set $\beta \coloneq \frac{2}{3 - \alpha}$.
If $\alpha < 1$, then $0 < \beta < 1$, and so by Lemma~\ref{lemma:nuuuu} we get
\begin{align*}
    C_{\alpha, K}
&=
     \lrb{2^{\frac{1}{\beta} - 1}}^{ 2}
= 
     {2}^{ \frac{2}{\beta} - 2}
= 
     {2}^{ 1 - \alpha}\;.
\end{align*}
If $\alpha =1$, then $\beta = 1$, and so $C_{1, K} = 1$.
If $\alpha \in (1,2]$, then $\beta > 1$, hence by Lemma~\ref{lemma:nuuuu} we get
\[
    C_{\alpha, K}
=
     2^{-2} \lrb{{\left\lfloor \fracc{K}{2} \right\rfloor}^{1 - \beta} +  {\left\lceil \fracc{K}{2} \right\rceil}^{1 - \beta}}^{2/\beta}\;,
\]
and for $K$ even this means that
\begin{align*}
    C_{\alpha, K}
&= 
     2^{-2} \lrb{\lrb{\frac{K}{2}}^{1 - \beta} +  \lrb{\frac{K}{2}}^{1 - \beta}}^{2/\beta}
= 
    2^{-2} \lrb{ 2 \cdot \frac{{K}^{1 - \beta}}{2^{1 - \beta}} }^{2/\beta}
= 
    \lrb{{K}^{1 - \beta}}^{2/\beta}
= 
    {K}^{ \frac{2}{\beta} - 2}
=
    {K}^{ 1 - \alpha}
\end{align*}
whereas for $K$ odd we have 
\begin{align*}
    C_{\alpha, K}
&= 
     2^{-2} \lrb{\lrb{\frac{K-1}{2}}^{1 - \beta} +  \lrb{\frac{K+1}{2}}^{1 - \beta}}^{2/\beta}
= 
     \lrb{  2^{- \beta} \cdot \frac{ (K-1)^{1 - \beta} + (K+1)^{1 - \beta} }{2^{1 - \beta} } }^{2/\beta}
\\
&= 
     \lrb{  \frac{ (K-1)^{\frac{1 - \alpha}{3- \alpha}} + (K+1)^{\frac{1 - \alpha}{3- \alpha}} }{2} }^{3 - \alpha}
= 
     {K}^{ 1 - \alpha}
    \lrb{  \frac{ (1-\frac{1}{K})^{\frac{1 - \alpha}{3- \alpha}} + (1+\frac{1}{K})^{\frac{1 - \alpha}{3- \alpha}} }{2} }^{3 - \alpha}\;.
\end{align*}

\mycase{$\alpha >2$, $K = 2$}
Note that $\tsphere{2} = \bcb{\pm (1/2, -1/2) }$, and so 
\[
    C_{\alpha, 2}
=
    \inf \lcb{\sum_{k \in [2]} v_k^2 \gamma_k^{\alpha -2} : v \in \tsphere{2}, \gamma \in \relint\lrb{\simplex{2}}}
=   
    \inf \lcb{\sum_{k \in [2]} \frac{1}{4}\, \gamma_k^{\alpha -2} : \gamma \in \relint\lrb{\simplex{2}}}\;.
\]
Consider the function $\gamma \mapsto f(\gamma) = \frac{1}{4}(\gamma_1^{\alpha - 2} + \gamma_2^{\alpha - 2})$
which is continuous on $\simplex{2}$. By density it follows that $C_{\alpha, 2} = \min \lcb{f(\gamma) : \gamma \in \simplex{2}}$.
If $\alpha \in (2,3)$, then the function $f$ is concave on $\simplex{2}$, and so the minimum is attained on some extreme point of $\simplex{2}$, leading to
\begin{align*}
        C_{\alpha , 2} = f\brb{(1,0)} = f\brb{(0,1)}
&=
        \frac{1}{4} (1^{\alpha - 2} + 0^{\alpha - 2})
=
        2^{-2}\;.
\end{align*}
If $\alpha = 3$, then  $f(\gamma)=\frac{1}{4}$ for every $\gamma \in \simplex{2}$, thus $C_{\alpha, 2} = \frac{1}{4}=2^{-2}$. 
If $\alpha \in (3,+\infty)$, then the function $f$ is convex on $\simplex{2}$, and its symmetry about $\gamma = (1/2, 1/2)$ entails that 
\begin{align*}
    C_{\alpha , 2}
&=
    f\brb{(1/2, 1/2)}
=
        \frac{1}{4} \lrb{ \lrb{\frac{1}{2}}^{\alpha - 2} + \lrb{\frac{1}{2}}^{\alpha - 2}}
=
        2^{1 - \alpha}\;.
\end{align*}
\mycase{$\alpha >2$, $K \geq 3$}
By hypothesis on $K$, there exists some $\bar{v} \in \tsphere{K}$ such that $\bar{v}_1 = 0$, and that we can use to define the function $\gamma \mapsto g(\gamma) \coloneq \sum_{k \in [K]} \bar{v}_k^2 \gamma_k^{\alpha -2}$.
In particular,
$ C_{\alpha, K} \le g(\gamma)$
for every $\gamma \in \relint{\brb{\simplex{K}}}$ and, by continuity of $g$ over $\simplex{K}$, this inequality holds for every $\gamma \in \simplex{K}$.
We readily see that $g(\bar{\gamma}) = 0$ if $\bar{\gamma} \in \simplex{K}$
is the vertex of the simplex such that 
$\bar{\gamma}_k = 1$ if and only if $k = 1$, and 
this means that $C_{\alpha, K} = 0$.
\end{proof}
}
\section{Conclusion}
We established a \emph{sharp} Pinsker-type inequality for the Bregman divergences induced by negative Tsallis entropies, with an explicit optimal constant $C_{\alpha,K}$ for every $(\alpha,K)$ and a geometric explanation of the phase transitions (including the sharp multiclass breakdown for $\alpha>2$ and the precise dimension/parity effects for $\alpha\in(1,2]$).
Beyond its information-theoretic interest, this inequality has immediate learning-theoretic consequences.
First, it gives a tight way to convert control on the excess risk of Tsallis (power) losses into $\lno{\cdot}_1$ (total-variation) control uniformly on predictive distributions over the simplex.
Second, via the standard plug-in argument, it yields a principled route from Tsallis surrogate performance to \emph{multiclass $0$--$1$ classification} regret bounds, clarifying exactly when the conversion is dimension-free and when it can degrade with $K$.
Third, in \emph{online learning}, the result identifies the optimal $\lno{\cdot}_1$-strong convexity (curvature) of Tsallis regularizers, sharpening the constants that enter Follow-the-Regularized-Leader/Mirror-Descent analyses and making explicit how the choice of $\alpha$ controls the geometry behind the algorithm.
\section*{Acknowledgments}

TC gratefully acknowledges the support of the University of Ottawa through grant GR002837 (Start-Up Funds) and that of the Natural Sciences and Engineering Research Council of Canada (NSERC) through grants RGPIN-2023-03688 (Discovery Grants Program) and DGECR2023-00208 (Discovery Grants Program, DGECR - Discovery Launch Supplement).

RC is partially supported by the MUR PRIN grant 2022EKNE5K (Learning in Markets and Society), funded by the NextGenerationEU program within the PNRR scheme, and the the EU Horizon CL4-2021-HUMAN-01 research and innovation action under grant agreement 101070617, project ELSA (European Lighthouse on Secure and Safe AI).
\bibliographystyle{plainnat}
%\bibliography{biblio}

\appendix
\section{Pinsker Inequality: Total Variation vs. \texorpdfstring{$\lno{\cdot}_1$}{L1 Norm }}
\label{s:TVvsEll1}

\begin{lemma}[Total variation vs.\ $\lno{\cdot}_1$]
Let $K\in\mathbb{N}$ with $K \ge 2$ and let $p,q\in\simplex{K}$ be probability distributions on $[K]$.
With
\[
    \TV{p}{q} \coloneqq \sup_{C\subset [K]} \labs{\sum_{k\in C}p(k)-\sum_{k\in C}q(k)},
\qquad
    \lno{p-q}_1 \coloneqq \sum_{k=1}^K \labs{p(k)-q(k)}\;,
\]
it holds that
\[
    \TV{p}{q}
=
    \frac{1}{2}\cdot \lno{p-q}_1\;.
\]
\end{lemma}

\begin{proof}
Let $d(k)\coloneqq p(k)-q(k)$ for $k\in[K]$. Since $p,q\in\simplex{K}$,
\[
    \sum_{k=1}^K d(k)
=
    \sum_{k=1}^K p(k)-\sum_{k=1}^K q(k)
=
    1-1
=
    0\;.
\]
Define the index sets
\[
    P\coloneqq \bcb{k\in[K]: d(k)\ge 0}\;,\qquad
    N\coloneqq \bcb{k\in[K]: d(k)<0}\;.
\]
For any $C\subset[K]$ we have
\[
    \sum_{k\in C} d(k)
\le
    \sum_{k\in C\cap P} d(k)
\le
    \sum_{k\in P} d(k)\;,
\]
so $\sup_{C\subset[K]} \sum_{k\in C} d(k)\le \sum_{k\in P} d(k)$, and equality is attained by
choosing $C=P$. Hence
\[
    \sup_{C\subset[K]} \sum_{k\in C} d(k)
=
    \sum_{k\in P} d(k)\;.
\]
Similarly, applying the same argument to $-d$ yields
\[
    \sup_{C\subset[K]} \sum_{k\in C} \brb{-d(k)}
=
    \sum_{k\in N} \brb{-d(k)}\;.
\]
But since $\sum_{k=1}^K d(k)=0$, we have
\[
    \sum_{k\in P} d(k)
=
    -\sum_{k\in N} d(k)
=
    \sum_{k\in N} \brb{-d(k)}\;.
\]
Therefore,
\[
    \TV{p}{q}
=
    \sup_{C\subset[K]}\labs{\sum_{k\in C} d(k)}
=
    \max \lrb{\sup_{C \subset [K]}\sum_{k\in C} d(k),\sup_{C \subset [K]}\sum_{k\in C} \lrb{-d(k)}}
=
\sum_{k\in P} d(k)\;.
\]
On the other hand,
\[
    \lno{p-q}_1
=
    \sum_{k=1}^K \labs{d(k)}
=
    \sum_{k\in P} d(k) + \sum_{k\in N} \brb{-d(k)}
=
    2 \cdot \sum_{k\in P} d(k)\;.
\]
Combining the last two displays gives $\TV{p}{q}=\tfrac12 \cdot \lno{p-q}_1$.
\end{proof}

\section{Tsallis Losses and Their Bayes and Excess Risk}
\label{app:loss_entropy_bregman}
In this section we prove that for any $\alpha \in \R$ the Bayes and excess risk associated to Tsallis (or power) losses $\ell_\alpha$ in \eqref{eq:tsallis_losses} (with $\ell_0(q,k) \coloneqq q_k^{-1} -K+\sum_{i \in [K]} \ln q_i$, and $\ell_1(q,k) \coloneqq -\ln q_k$) are respectively the Tsallis entropy $S_\alpha$ (Definition~\ref{def:tsallis}) and their Bregman divergences $D_\alpha$ (Definition~\ref{def:bregman}).
Specifically, if $\alpha \in \R \backslash \lcb{0,1}$, for any $p,q \in \relint\lrb{\simplex{K}}$, we have that
\begin{align*}   
    \E_{Y \sim p}\lsb{ \ell_\alpha(q,Y)}
&=
    \sum_{k\in[K]} \ell_\alpha(q,k)\cdot p_k
=
    \sum_{k \in [K]} \lrb{ \frac{q^{\alpha-1}_k}{1-\alpha}+\frac{1}{\alpha}\sum_{i \in [K]}q_i^\alpha}\cdot p_k
\\
&=
    \sum_{k \in [K]} \frac{\alpha \cdot p_k \cdot q^{\alpha-1}_k+(1-\alpha)q_k^\alpha}{\alpha(1-\alpha)}
\end{align*}
which, for the Bayes risk case $p = q$, becomes
\[
    \E_{Y \sim p}\lsb{ \ell_\alpha(p,Y)}
=
    \sum_{k \in [K]} \frac{\alpha \cdot p_k \cdot p^{\alpha-1}_k+(1-\alpha)p_k^\alpha}{\alpha(1-\alpha)}
=
    \sum_{k \in [K]} \frac{p_k^\alpha}{\alpha(1-\alpha)}
=
    S_\alpha(p)\;,
\]
while the excess risk becomes
\begin{align*}    
    \E_{Y \sim p}\lsb{ \ell_\alpha(q,Y)}-\E_{Y \sim p}\lsb{ \ell_\alpha(p,Y)}
&=
    \sum_{k \in [K]} \frac{\alpha \cdot p_k \cdot q^{\alpha-1}_k+(1-\alpha)q_k^\alpha}{\alpha(1-\alpha)}
    -
    \sum_{k \in [K]} \frac{p_k^\alpha}{\alpha(1-\alpha)}
\\
&=
    \frac{1}{\alpha} \cdot \sum_{k \in [K]} \frac{p_k^\alpha + (\alpha-1)q_k^\alpha-\alpha p_kq_k^{\alpha-1}}{\alpha-1}
=
    d_\alpha(p \Vert q)
=
    D_\alpha(p \Vert q)\;,
\end{align*} 
where $d_\alpha$ is the $\beta$-divergence for $\beta = \alpha$ as defined in Definition~\ref{def:beta}, and the last inequality follows from Remark~\ref{r:beta-divergences}.

For the case $\alpha = 0$, if $p,q \in \relint\lrb{\simplex{K}}$, we have
\begin{align*}    
    \E_{Y \sim p}\lsb{ \ell_0(q,Y)}
&=
    \sum_{k\in[K]} \ell_0(q,k)\cdot p_k
=
    \sum_{k\in[K]} \lrb{q_k^{-1} - K + \sum_{i \in [K]} \ln q_i}\cdot p_k
=
    \sum_{k\in[K]} \lrb{\frac{p_k}{q_k} -1 +  \ln q_k}
\end{align*}
which, for the Bayes risk case $p = q$, becomes
\[
    \E_{Y \sim p}\lsb{ \ell_0(p,Y)}
=
    \sum_{k\in[K]} \lrb{\frac{p_k}{p_k}  -1 + \ln p_k}
=
    \sum_{k\in[K]} \ln p_k
=
    S_0(p)
\]
while the excess risk becomes
\begin{align*}
    \E_{Y \sim p}\lsb{ \ell_0(q,Y)}-\E_{Y \sim p}\lsb{ \ell_0(p,Y)}
&=
    \sum_{k\in[K]} \lrb{\frac{p_k}{q_k} -1 +  \ln q_k}-\sum_{k\in[K]} \ln p_k
\\
&=
    \sum_{k=1}^K
    \lrb{ \frac{p_k}{q_k} -\ln\frac{p_k}{q_k} -1 }
=
    d_0(p\Vert q)
=
    D_0(p\Vert q)\;,
\end{align*}
where $d_0$ is the $\beta$-divergence for $\beta = 0$ (or Itakura-Saito divergence) as defined in Definition~\ref{def:beta}, and the last inequality follows from Remark~\ref{r:beta-divergences}.

For the case $\alpha = 1$, if $p,q \in \relint\lrb{\simplex{K}}$, we have
\begin{align*}    
    \E_{Y \sim p}\lsb{ \ell_1(q,Y)}
&=
    \sum_{k\in[K]} \ell_1(q,k)\cdot p_k
=
    \sum_{k\in[K]} \lrb{ - \ln q_k}\cdot p_k
=
    -\sum_{k\in[K]} p_k \cdot \ln q_k\;,
\end{align*}
which, for the Bayes risk case $p = q$, becomes
\[
    \E_{Y \sim p}\lsb{ \ell_1(p,Y)}
=
    -\sum_{k\in[K]} p_k \cdot \ln p_k
=
    S_1(p)
\]
while the excess risk becomes
\begin{align*}
    \E_{Y \sim p}\lsb{ \ell_1(q,Y)}-\E_{Y \sim p}\lsb{ \ell_1(p,Y)}
&=
    -\sum_{k\in[K]} p_k \cdot \ln q_k + \sum_{k\in[K]} p_k \cdot \ln p_k
\\
&=
    \sum_{k \in [K]} p_k \cdot \ln \frac{p_k}{q_k}
=
    \KL(p\Vert q)
=
    D_1(p\Vert q)\;,
\end{align*}
where $\KL$ is the Kullback-Leibler divergence, and the last inequality follows from Remark~\ref{r:kullback-leibler}.
\section[0--1 Loss Regret is Controlled by L1 Distance]{$0$--$1$ Loss Regret is Controlled by $\lno{\cdot}_1$ distance}
\label{app:0-1}
\begin{lemma}[$0$-$1$ loss regret is controlled by $\lno{\cdot}_1$ distance]
\label{lem:01_from_l1}
For any $p,q\in\simplex{K}$, select $k^\star_p\in\argmax_{k \in [K]} p_k$ and $k^\star_q \in\argmax_k q_k$, deterministically.
Then
\begin{align*}
    \Pb_{Y \sim p}\lsb{Y \neq k^\star_q}
    -
    \Pb_{Y \sim p}\lsb{Y \neq k^\star_p}
\le
    \lno{p-q}_1\;.
\end{align*}
\end{lemma}
\begin{proof}
If $k^\star_q = k^\star_p$ there is nothing to prove.
Assume then $k^\star_q \neq k^\star_p$.
Then, we have
\begin{align*}
    \Pb_{Y \sim p}\lsb{Y \neq k^\star_q}
    -
    \Pb_{Y \sim p}\lsb{Y \neq k^\star_p}
&=
    1-p_{k_q^\star}-(1-p_{k_p^\star})
=
    p_{k_p^\star}-p_{k_q^\star}
\\
&=
    p_{k_p^\star}-q_{k_p^\star}
    +\underbrace{q_{k_p^\star}-q_{k_q^\star}}_{\le 0}
    +q_{k_q^\star}-p_{k_q^\star}
\\
&\le
    p_{k_p^\star}-q_{k_p^\star}
    +q_{k_q^\star}-p_{k_q^\star}
\\
&=
    \labs{p_{k_p^\star}-q_{k_p^\star}}
    +
    \labs{p_{k_q^\star}-q_{k_q^\star}}
\le
    \sum_{k=1}^K \labs{p_k-q_k}
=
    \lno{p-q}_1\;.
\end{align*}
\end{proof}
\section{Proof of the Remarks}
\label{app:remarks}
{\renewcommand{\proofname}{Proof of Remark \ref{r:beta-divergences}}\begin{proof}
For any $\alpha \in \R$, the Bregman divergence $D_\alpha$ of $-S_\alpha$ coincides with the \emph{$\beta$-divergence} $d_\beta$, with $\beta = \alpha$.
Indeed, for any $p,q \in (0,+\infty)^K$ we have
\begin{align*}
    D_{\alpha}(p\Vert q)
&=
    -\frac{\sum_{k \in[K]} p_k^\alpha}{\alpha(1-\alpha)} + \frac{\sum_{k \in[K]} q_k^\alpha}{\alpha(1-\alpha)}  - \sum_{k \in[K]} \frac{1}{\alpha-1} q_k^{\alpha-1}(p_k-q_k)
\\
&=
    \frac{1}{\alpha} \cdot \sum_{k \in [K]} \frac{p_k^\alpha + (\alpha-1)q_k^\alpha-\alpha p_kq_k^{\alpha-1}}{\alpha-1}
=
    d_\alpha(p \Vert q)\;, \qquad \alpha \notin \{0,1\}
\end{align*}
while for $\alpha = 0$ we have
\[
    D_{0}(p \Vert q)
=
    -\sum_{k \in [K]} \ln p_k
    +\sum_{k \in [K]} \ln q_k
    +\sum_{k \in [K]} \frac{p_k-q_k}{q_k}
=
    \sum_{k=1}^K
    \lrb{ \frac{p_k}{q_k} -\ln\frac{p_k}{q_k} -1 }
=
    d_0(p\Vert q)\;,
\]
which is also known as the Itakura-Saito divergence, and for $\alpha = 1$ we have
\begin{align*}    
    D_{1}(p\Vert q)
&=
    \sum_{k \in [K]} p_k \ln p_k - \sum_{k \in [K]} q_k \ln q_k - \sum_{k \in [K]} (1+\ln q_k)(p_k-q_k)
\\
&=
    \sum_{k \in [K]} \lrb{ p_k \cdot\ln\frac{p_k}{q_k} - p_k + q_k}
=
    d_1(p\Vert q)\;,
\end{align*}

Now, we prove the continuity statement.
The only two non-trivial cases are $\alpha \in \lcb{0,1}$.
For the case $\alpha = 0$, for any $p,q \in (0,+\infty)^K$, noticing that the limit is an indeterminate $0/0$ form, we can apply l'Hôpital's rule $(\mathrm{H})$ to obtain
\begin{align*}
    \lim_{\alpha \to 0}D_\alpha(p \Vert q)
&\overset{\phantom{(\mathrm{H})}}{=}
    \lim_{\alpha \to 0}\frac{1}{\alpha} \cdot \sum_{k \in [K]} \frac{p_k^\alpha + (\alpha-1)q_k^\alpha-\alpha p_kq_k^{\alpha-1}}{\alpha-1}
\\
&\overset{\phantom{(\mathrm{H})}}{=}
    \sum_{k \in [K]} \lim_{\alpha \to 0} \frac{-p_k^\alpha + (1-\alpha)q_k^\alpha+\alpha p_kq_k^{\alpha-1}}{\alpha}
\\
&\overset{(\mathrm{H})}{=}
    \sum_{k \in [K]} \lim_{\alpha \to 0} \lrb{-p_k^\alpha \ln p_k -q_k^\alpha+(1-\alpha)q_k^\alpha \ln q_k+p_kq_k^{\alpha-1}+\alpha p_kq_k^{\alpha-1} \ln q_k}
\\
&\overset{\phantom{(\mathrm{H})}}{=}
    \sum_{k \in [K]} \lrb{ -\ln p_k -1+\ln q_k+p_kq_k^{-1}}
\\
&\overset{\phantom{(\mathrm{H})}}{=}
    \sum_{k \in [K]}\lrb{ \frac{p_k}{q_k} -\ln\frac{p_k}{q_k} -1 }
=
    D_0(p\Vert q)\;.
\end{align*}
For the case $\alpha = 1$, for any $p,q \in (0,+\infty)^K$, noticing that the limit is again an indeterminate $0/0$ form,
we can apply l'Hôpital's rule $(\mathrm{H})$ to obtain
\begin{align*}
    \lim_{\alpha \to 1}D_\alpha(p \Vert q)
&\overset{\phantom{(\mathrm{H})}}{=}
    \lim_{\alpha \to 1}\frac{1}{\alpha} \cdot \sum_{k \in [K]} \frac{p_k^\alpha + (\alpha-1)q_k^\alpha-\alpha p_kq_k^{\alpha-1}}{\alpha-1}
\\
&\overset{\phantom{(\mathrm{H})}}{=}
    \sum_{k \in [K]} \lim_{\alpha \to 1} \frac{p_k^\alpha + (\alpha-1)q_k^\alpha-\alpha p_kq_k^{\alpha-1}}{\alpha-1}
\\
&\overset{(\mathrm{H})}{=}
    \sum_{k \in [K]} \lim_{\alpha \to 1} \lrb{p_k^\alpha \ln p_k +q_k^\alpha+(\alpha-1)q_k^\alpha \ln q_k-p_kq_k^{\alpha-1}-\alpha p_kq_k^{\alpha-1} \ln q_k}
\\
&\overset{\phantom{(\mathrm{H})}}{=}
    \sum_{k \in [K]} \lrb{p_k \ln p_k +q_k-p_k-p_k \ln q_k}
\\
&\overset{\phantom{(\mathrm{H})}}{=}
    \sum_{k \in [K]} \lrb{ p_k \cdot\ln\frac{p_k}{q_k} - p_k + q_k}
=
    D_1(p\Vert q)\;.
\end{align*}
\end{proof}
}
{\renewcommand{\proofname}{Proof of Remark \ref{rem:nopinsker-orthant}}\begin{proof}
For $\alpha \neq 2$ the remark is equivalent to showing that
\[
    \inf_{\substack{p,q\in(0,+\infty)^K\\ p\neq q}}\ \frac{2D_\alpha(p\Vert q)}{\lno{p-q}_1^2}
=
    0\;.
\]
Let $\mathbf{1}\in\R^K$ be the all-ones vector and $e_1$ the first canonical basis vector.
Using a second-order Taylor expansion of $-S_\alpha$ at $q$ and the Hessian expression
$
    H(-S_\alpha)(q)
=
    \mathrm{diag}(q_1^{\alpha-2},\dots,q_K^{\alpha-2})
$
(Remark~\ref{remark:hessian}),
we obtain for $q\in(0,+\infty)^K$ and $\e\to 0^+$ that
\[
    D_\alpha(q+\e \cdot e_1 \Vert q)=\frac{\e^2}{2}\cdot q_1^{\alpha-2}+o(\e^2)\;.
\]
Now take $q(t)\coloneqq t\cdot\mathbf{1}$, so $q_1(t)=t$, and note that $\lno{ \brb{q(t)+\e \cdot e_1}-q(t)}_1=\e$.
Hence
\[
    \varphi_\alpha(\e,t)
\coloneqq
    \frac{2D_\alpha\brb{q(t)+\e \cdot e_1 \Vert q(t)}}{\lno{\brb{q(t)+\e \cdot e_1}-q(t)}_1^2}
=
    t^{\alpha-2}+o(1)\;,\qquad \textrm{as $\e \to 0^+$}\;.
\]
If $\alpha<2$, set $\e \coloneqq \frac{1}{t}$ and send $t\to+\infty$, so that $\varphi_\alpha\lrb{\frac{1}{t},t} \to 0$.
If $\alpha>2$, set $\e \coloneqq t$ and send $t\to0^+$, so that  $\varphi_\alpha\lrb{t,t} \to 0$.
This proves the claim.

\noindent
In contrast, when $\alpha=2$, as we already noted in the footnote to Remark~\ref{r:euclidean}, we have that
\[
    D_2(p\Vert q)
=
    \frac12 \cdot \lno{p-q}_2^2 \qquad \text{for all }p,q\in(0,+\infty)^K\;.
\]
Therefore, since from Jensen's inequality we have, for any $x \in \R^K$,
\[
    \lno{x}_2
=
    \lrb{\sum_{k \in [K]} x_k^2}^{1/2}
=
    \sqrt{K} \cdot\lrb{\frac{1}{K}\sum_{k \in [K]} x_k^2}^{1/2}
\ge
    \frac{1}{\sqrt{K}}\sum_{k=1}^K|x_k|
=
    \frac{1}{\sqrt{K}}\cdot \lno{x}_1\;,
\]
we get as well
\[
    D_2(p\Vert q)
=
    \frac12 \cdot \lno{p-q}_2^2
\ge
    \frac{1}{2K}\cdot \lno{p-q}_1^2,
\qquad
    p,q\in(0,+\infty)^K\;,
\]
with the equality attained at $q \coloneqq p + c\cdot \mathbf{1}$, for any $c >0$ and any $p \in (0,+\infty)^K$.
\end{proof}
}
{\renewcommand{\proofname}{Proof of Remark \ref{remark:hessian}}
\begin{proof}
If $\alpha \in (-\infty,0) \cup (0,1) \cup (1,+\infty)$, for any $i,j \in [K]$, we have that
\[
    \partial_i S_\alpha(p)=-\frac{1}{\alpha-1}p_i^{\alpha-1},
\]
from which it immediately follows that, $\partial_j\partial_i S_\alpha = 0$ whenever $j \neq i$, while
\[
    \partial_{ii}^2 S_\alpha(p)
=
    \partial_i \lrb{-\frac{1}{\alpha-1}p_i^{\alpha-1}}
=
    -p_i^{\alpha-2}\;.
\]

If $\alpha = 0$, we have $S_0(p)=\sum_{k\in[K]}\ln p_k$, so for any $i\in[K]$,
\[
\partial_i S_0(p)=\frac{1}{p_i},\qquad
\partial^2_{ii}S_0(p)=-\frac{1}{p_i^2},\qquad
\partial^2_{ij}S_0(p)=0\ \ (i\neq j).
\]
Since $p_i^{-2}=p_i^{\alpha-2}$ when $\alpha=0$, this matches the claimed formula.

If $\alpha =1$, we have $S_1(p)=-\sum_{k\in[K]}p_k\ln p_k$, so for any $i\in[K]$,
\[
\partial_i S_1(p)=-(1+\ln p_i),\qquad
\partial^2_{ii}S_1(p)=-\frac{1}{p_i},\qquad
\partial^2_{ij}S_1(p)=0\ \ (i\neq j).
\]
Since $p_i^{-1}=p_i^{\alpha-2}$ when $\alpha=1$, this again matches the claimed formula.

Putting the cases together, for every $\alpha \in \R$ and $p\in(0,+\infty)^K$,
\[
\partial^2_{ij} S_\alpha(p)= -\,p_i^{\alpha-2}\,\I\{i=j\},
\]
and thus $H S_\alpha(p)=-\mathrm{diag}(p_1^{\alpha-2},\dots,p_K^{\alpha-2})$.
\end{proof}
}
{\renewcommand{\proofname}{Proof of Remark \ref{remark:const-is-best}}
\begin{proof}
Given $C \geq 0$ and $p$, $q \in  \relint\lrb{\simplex{K}}$, the inequality
\begin{equation}\label{eq:pinsker-C}
    D_\alpha\brb{p\Vert q}
\ge
    \frac{C}{2} \cdot \lno{p - q}_1^2    
\end{equation}
is trivial if $p = q$, whereas for $p \neq q$ it is verified if and only if
\[
    2 \cdot\frac{D_\alpha\brb{p\Vert q}}{\lno{p - q}_1^2}
\ge
    C\;.
\]
From this we see that $C \geq 0$ satisfies \eqref{eq:pinsker-C} for \emph{every} $p$, $q \in  \relint\lrb{\simplex{K}}$ if and only if 
\begin{equation}\label{eq:ratio-inf}    
    2 \cdot \inf \lcb{\frac{D_\alpha\brb{p\Vert q}}{\lno{p - q}_1^2} : p , q \in  \relint\lrb{\simplex{K}}, p \neq q}
\ge
    C\;,
\end{equation}
and the largest such $C$ equals the LHS of \eqref{eq:ratio-inf}.
\end{proof}
}
\subsection[Proof of Remark Sigma: Estimating the Constant sigma(alpha,K)]{Proof of Remark \ref{rem:sigma}: Estimating the Constant $\sigma_{\alpha,K}$}\label{app:sigma}
\begin{proposition}\label{prop:sigma}
Let $\alpha \in (1,2)$ and $K \geq 3$.
Then
 \[
     \sigma_{\alpha, K}
\coloneq
    \lrb{  \frac{ (1-\frac{1}{K})^{\frac{1 - \alpha}{3- \alpha}} + (1+\frac{1}{K})^{\frac{1 - \alpha}{3- \alpha}} }{2} }^{3 - \alpha}\;
\]
satisfies
\begin{equation}\label{eq:sigma-bounds}    
    1 + \frac{\alpha - 1}{3- \alpha} \cdot 
        \frac{1}{K^2} 
\le
    \sigma_{\alpha, K}
\le
    1 + \frac{7(\alpha - 1)}{6(3- \alpha)} \cdot 
        \frac{1}{K^2}\;.
\end{equation}
\end{proposition}
\begin{proof}
Define $\nu \coloneq (1 - \alpha)/(3- \alpha)$ and note that $\nu \in (-1,0)$.
Next, define
\[
    z(\alpha, K)
\coloneq
    \frac{ (1-\frac{1}{K})^{\nu} + (1+\frac{1}{K})^{\nu} }{2} - 1\;,
\]
which, by considering the binomial series for $(1 + x)^\nu$ and $x = \pm 1/K$, can be expanded as
\begin{equation}\label{eq:z-expasion}
    z(\alpha, K)
=
    \sum_{n=1}^{\infty} \binom{\nu}{2n} \frac{1}{K^{2n}}\;,    
\end{equation}
where the coefficients are nonnegative, as shown by the recurrence 
\begin{equation}\label{eq:binom-recurrence}    
    \binom{\nu}{2n+2} 
= 
    \binom{\nu}{2n}
    \cdot
    \frac{(\nu- 2n)(\nu - 2n - 1)}{ (2n + 1)(2n + 2) }\;.
\end{equation}
In particular, the function $K \mapsto z(\alpha, K)$ is an increasing function of $1/K^2$, and 
\[
    z(\alpha, K)
\ge
    \binom{\nu}{2}\cdot  \frac{1}{K^{2}}
=
    \frac{\alpha - 1}{(3- \alpha)^2} \cdot \frac{1}{K^2}\,.
\]
This inequality, together with Bernoulli's inequality, yields
\[
    1 +  \frac{\alpha - 1}{(3- \alpha)} \cdot \frac{1}{K^2}
\le
    1 + (3-\alpha)\cdot z(\alpha, K)
\le
    (1 + z(\alpha, K))^{3-\alpha}
=   
    \sigma_{\alpha, K}\;,
\]
and this proves the lower bound in \eqref{eq:sigma-bounds}.
For the upper bound, first note that $K \ge 3$ and that
$\binom{\nu}{2n}$ is decreasing in $n$ by \eqref{eq:binom-recurrence}, hence
\[
    z(\alpha, K)
\le
    \binom{\nu}{2}\cdot 
    \sum_{n=1}^{\infty}  \frac{1}{K^{2n}}
\le
    \binom{\nu}{2}\cdot 
    \frac{1}{K^2}\cdot 
    \sum_{n=0}^{\infty}  \frac{1}{3^{2n}}
=
    \frac{9(\alpha - 1)}{8(3- \alpha)^2}
    \cdot \frac{1}{K^2}\,.
\]
By Bernoulli's inequality,
\begin{align*}
    \sigma_{\alpha, K}
&=
    \lrb{1 + z(\alpha, K)}
    \cdot (1 + z(\alpha, K))^{2-\alpha}
\\
&\le
    (1 + z(\alpha, K))\cdot (1 + (2- \alpha)\cdot z(\alpha, K))
\\
&=
    1 + (3 - \alpha)\cdot  z(\alpha, K) \cdot 
        \lrb{1 +  \frac{2 - \alpha}{3 - \alpha} \cdot z(\alpha, K) }
\;,
\end{align*}
hence
\begin{align*}
    \sigma_{\alpha, K}
&\le
    1 + (3 - \alpha) \cdot z(\alpha, K) \cdot 
        \lrb{1 +  \frac{2 - \alpha}{3 - \alpha} \cdot z(\alpha, 3)}
\\
&\le
    1 +  \frac{9(\alpha - 1)}{8(3- \alpha)}
           \cdot \frac{1}{K^2}\cdot 
        \lrb{1 +  (2 - \alpha)\frac{9(\alpha - 1)}{8(3- \alpha)^3} \frac{1}{3^2}}
\\
&=
    1 +  \frac{9(\alpha - 1)}{8(3- \alpha)} 
        \cdot \lrb{1 +  \frac{(2-\alpha)(\alpha - 1)}{8(3- \alpha)^3}}
        \cdot \frac{1}{K^2}
\\
&\le
    1 +  \frac{9(\alpha - 1)}{8(3- \alpha)} 
        \cdot \lrb{1 +  \frac{1}{32}}
        \cdot \frac{1}{K^2}
\\
&=
    1 +  \frac{297(\alpha - 1)}{256(3- \alpha)} 
        \cdot \frac{1}{K^2}\;,
\end{align*}
using that $1/32$ is an upper bound of $(2-\alpha)(\alpha - 1)/\brb{8(3- \alpha)^3}$ for $\alpha \in (1, 2)$.
Noticing that we have $\frac{297}{256} \le \frac{7}{6}$ we obtain the conclusion.
\end{proof}
The inequality given in Remark~\ref{rem:sigma} is based on Proposition~\ref{prop:sigma} and on the observation that $0 \le \frac{\alpha - 1}{3-\alpha} \le 1$ for $\alpha \in (1,2]$.
\section{Nonexistence of Pinsker-Type Inequality for \texorpdfstring{$\alpha >2$ and $K \geq 3$} {alpha > 2 and K >= 3} }\label{app:nopinsker}
It is possible to show directly that for $\alpha >2$ and  $K \geq 3$ the divergence $D_\alpha$ does not admit a Pinsker-type inequality, i.e., that it is not possible to find a \emph{positive} constant $C$ such that $D_\alpha(p\Vert q) \ge C\cdot 
    \lno{ p-q }_1^2$ for every $p$, $q \in \relint{\lrb{\simplex{K}}}$.
\begin{theorem}
\label{t:K_big_alpha_big}
Let $K \geq 3$ be an integer. For every $\alpha \in (2, +\infty)$ and every $\e >0$, there exist $p,q \in \relint\lrb{\simplex{K}}$ such that
\[
    D_\alpha(p\Vert q)
<
    \e \cdot \lno{ p-q }_1^2\;.
\]
\end{theorem}
\begin{proof}
Define 
\begin{align*}   
    p_k(t) 
&
\coloneq 
    \begin{cases}
            1 - (K-1) \cdot t
        &
            \text{for $k = 1$}
        \\
            \frac{3}{4}t 
        &
            \text{for $k = 2$}
        \\
            \frac{5}{4}t 
        &
            \text{for $k = 3$}
        \\
            t
        &
            \text{for $3 <k \le K$}
    \end{cases}
&
    q_k(t)
& \coloneq 
    \begin{cases}
            1 - (K-1) \cdot t
        &
            \text{for $k = 1$}
        \\
            \frac{5}{4}t 
        &
            \text{for $k = 2$}
        \\
            \frac{3}{4}t 
        &
            \text{for $k = 3$}
        \\
            t
        &
            \text{for $3 <k \le K$.}
    \end{cases}
\end{align*}

For $0 < t <\frac{1}{K-1}$, both $p \coloneq p(t)$ and $q \coloneqq q(t)$ are in $\relint\lrb{\simplex{K}}$.
By construction, $S_{\alpha}\brb{p(t)} = S_{\alpha}\brb{q(t)}$, hence
\begin{align*}
    D_\alpha\brb{p\Vert q}
&=
    \langle\nabla S_{\alpha}(q),p-q\rangle
\\
&=
    - \sum_{k \in[K]} \frac{1}{\alpha-1} \cdot  q_k^{\alpha-1} \lrb{p_k-q_k}
\\
&=
    - \frac{1}{\alpha-1} \lrb{ 
            \lrb{\frac{5}{4} \cdot  t }^{\alpha-1} \cdot  \lrb{- \frac{1}{2} \cdot t}
        +
            \lrb{\frac{3}{4} \cdot t }^{\alpha-1} \cdot \lrb{\frac{1}{2} \cdot t}
        }
\\
&=
    t^{\alpha} 
    \cdot 
    \frac{1}{2(\alpha-1)}
    \cdot
    \frac{5 ^{\alpha-1} -   {3}^{\alpha - 1}}{4^{\alpha-1}}\;,
\end{align*}
and note that $\lno{p - q}_1 = t$, thus
\begin{align*}
    \frac{D_\alpha\brb{p\Vert q}}{ \lno{p - q}_1^2}
&
=
    t^{\alpha - 2} 
\cdot 
    \frac{1}{2(\alpha-1)} 
\cdot 
    \frac{5 ^{\alpha-1} -   {3}^{\alpha - 1}}{4^{\alpha-1}}\;.
\end{align*}
By hypothesis on $\alpha$ we have $t^{\alpha - 2} \to 0$ as $t \to 0^+$, thus the thesis follows by considering, for every $\e >0$, a sufficiently small $t = t(\e)$.
\end{proof}
{\renewcommand{\proofname}{Proof}

We underscore that the construction provided in the proof of Theorem~\ref{t:K_big_alpha_big} requires $K \geq 3$. In fact, the same argument fails for $K = 2$.
\section{Pinsker Inequality for Csisz\'ar \texorpdfstring{$f$}{f}-Divergences}
\label{sec:f-divergences}
In this section we recall key definitions and results on Csisz\'ar $f$-divergences (in particular, $\alpha$-Tsallis relative entropies), and we compare the classical generalized Pinsker inequality for $f$-divergences with our generalized Pinsker inequality for the Bregman divergence induced by the negative $\alpha$-Tsallis entropies.

\begin{definition}
We say that $f$ is an admissible generator if $f \colon [0,+\infty) \to (-\infty,+\infty]$ is a convex function such that $f$ is finite for every $x >0$, $f(1)=0$ and $f(0) = \lim_{x \to 0^+}f(x)$.
If $f$ is an admissible generator, and $\Pb$ and $\Q$ are probability measures both defined over some measurable space $(\Omega,\cF)$ and such that $\Pb\ll\Q$ (which means that $\Pb$ is absolutely continuous with respect to $\Q$, i.e., that $\Q[A] = 0$ implies $ \Pb[A]=0$, for every event $A \in \cF$), the Csisz\'ar $f$-divergence of $\Pb$ from $\Q$ is defined by
\[
    D_f^{\textrm{C}}(\Pb,\Q)
\coloneqq
    \int_\Omega f\lrb{\frac{\dif \Pb}{\dif \Q}}\dif \Q\;,
\]
where $\frac{\dif \Pb}{\dif \Q}$ is the Radon-Nikodym derivative of $\Pb$ with respect to $\Q$.
\end{definition}

For any admissible generator $f$ that is three times differentiable at $1$ with $f''(1)>0$, and for any probability measures $\Pb$ and $\Q$ on some measurable space $(\Omega,\cF)$ such that $\Pb\ll\Q$, it has been proven by \cite{Gilardoni} that a generalized Pinsker inequality holds (and is sharp) in the following form
\begin{equation}
\label{eq:Pinsker_Csiszar}
   D_f^{\textrm{C}}(\Pb,\Q)
\ge
    2f''(1) \cdot \TV{\Pb}{\Q}^2\;,
\end{equation}
where $\TV{\Pb}{\Q}$ is the total variation distance between $\Pb$ and $\Q$, formally defined as
\[
    \TV{\Pb}{\Q}
\coloneqq
    \sup_{A \in \cF} \babs{ \Pb[A] - \Q[A]}\;.
\]

The KL-divergence $\KL(\Pb\Vert \Q)$ between two probability measures $\Pb$ and $\Q$ on some measurable space $(\Omega,\cF)$ such that $\Pb\ll\Q$ is indeed a Csisz\'ar $f$-divergence for the admissible generator $f \colon [0,+\infty) \to (-\infty,+\infty]$ given by $x \mapsto f(x)\coloneq x \ln x$ (with the convention that $0 \cdot \ln0 = 0$), since, by a change of variables, we have
\[
    \KL(\Pb\Vert \Q)
\coloneqq
    \int_{\Omega} \ln \lrb{\frac{\dif \Pb}{\dif \Q}} \dif \Pb
=
    \int_{\Omega} \ln \lrb{\frac{\dif \Pb}{\dif \Q}} \frac{\dif \Pb}{\dif \Q} \dif \Q
=
    \int_{\Omega} f\lrb{\frac{\dif \Pb}{\dif \Q}} \dif \Q
=
D_f^{\textrm{C}}(\Pb,\Q)\;,
\]
and, consequently, for \Cref{eq:Pinsker_Csiszar}, we recover the (sharp) Pinsker inequality for probability measures
\[
    \KL(\Pb\Vert \Q)
=
    D_f^{\textrm{C}}(\Pb,\Q)
\ge
    2 f''(1) \cdot \TV{\Pb}{\Q}^2 
=
    2 \cdot \TV{\Pb}{\Q}^2\;.
\]
For $K \in \{2,3,\dots\}$, this last inequality implies the classic Pinsker inequality over the $(K-1)$-dimensional probability simplex $\simplex{K}$.
Specifically, if $p,q \in \simplex{K}$ are such that, for any $k \in [K]$, whenever $q_k = 0$ it also holds that $p_k = 0$, then $p$ and $q$ are the probability mass functions of two probability measures $\Pb$ and $\Q$ on the discrete space $\Omega \coloneqq [K]$ with $\Pb\ll\Q$, and
\begin{align*}
    \TV{\Pb}{\Q}
&=
    \sup_{A \in 2^K} \babs{\Pb[A]-\Q[A]}
=
    \sup_{A \in 2^K} \labs{\sum_{k\in A}p_k-\sum_{k\in A}q_k}
\\
&=
    \frac{1}{2} \cdot \sum_{k \in [K]}|p_k-q_k|
=
    \frac{1}{2}\cdot \lno{p-q}_1
=
    \TV{p}{q}\;,
\end{align*}
(the last equality is proven in \Cref{s:TVvsEll1}), which implies (with the convention $0/0 \coloneqq 0$ and $0 \cdot \ln(0) \coloneqq 0$)
\[
    \KL(p\Vert q)
\coloneqq
    \sum_{k \in [K]} \ln\lrb{\frac{p_k}{q_k}}p_k
=
    \int_{\Omega} \ln\lrb{\frac{\dif \Pb}{ \dif \Q}} \dif\Pb
=
    \KL(\Pb \Vert \Q)
\ge
    2 \cdot\TV{\Pb}{\Q}^2
=
    2 \cdot\TV{p}{q}^2\;.
\]

An interesting class of admissible generators are those of the form
\[
    f_\alpha \colon [0,+\infty) \to (-\infty,+\infty]\;,\qquad
    x\mapsto \frac{x^{\alpha}-1}{\alpha(\alpha-1)}\;,
\]
for $\alpha \in (0,1)\cup(1,+\infty)$, whose corresponding Csisz\'ar $f_\alpha$-divergences go under the name of $\alpha$-Tsallis relative entropies (up to perhaps different normalization constants), with \Cref{eq:Pinsker_Csiszar} that becomes
\[
    D^{\textrm{TRE}}_\alpha(\Pb \Vert \Q)
\coloneqq
     D^{\textrm{C}}_{f_\alpha}(\Pb , \Q)
\ge
    2 f_\alpha '' (1) \cdot \TV{\Pb}{\Q}^2
=
    2\cdot \TV{\Pb}{\Q}^2\;.
\]
Again, if we specialize this last inequality in the discrete case, for any $p,q \in \relint\lrb{\simplex{K}}$, if $\Pb$ and $\Q$ are the probability measures of which $p$ and $q$ are the respective probability mass functions, we find that
\begin{align}
\label{eq:pinsker_TRE}
    D^{\textrm{TRE}}_\alpha(p \Vert q)
&\coloneqq
    \sum_{k \in [K]} f_\alpha\lrb{\frac{p_k}{q_k}}q_k
=
    \int_\Omega f\lrb{\frac{\dif \Pb}{\dif \Q}}\dif \Q
=
    D^{\textrm{TRE}}_\alpha(\Pb \Vert \Q)
\\
&\ge
    2 \cdot\TV{\Pb}{\Q}^2
=
    2 \cdot \TV{p}{q}^2
=
    \frac{1}{2} \cdot \lno{p-q}_1^2\;.\nonumber
\end{align}
It is interesting to investigate the relation between $\alpha$-Tsallis relative entropies and Bregman divergences of negative $\alpha$-Tsallis entropies, their respective generalized Pinsker inequalities, and whether they can be derived from each other.
Here, we show that for $\alpha \in (-\infty,0)\cup (0,1)$, the generalized Pinsker inequality for $\alpha$-Tsallis relative entropies of \cite{Gilardoni} (Equation~\ref{eq:pinsker_TRE}) implies a \emph{looser} version of our Pinsker inequality for Bregman divergences of negative $\alpha$-Tsallis entropies (\Cref{thm:pinsker}), while the vice versa is true for $\alpha \in (1,+\infty)$.

Specifically, for any $\alpha \in (-\infty,0)\cup (0,1) \cup (1,+\infty)$ and any $p,q \in \relint(\simplex{K})$, we have
\begin{align*}
    D_\alpha(p \Vert q)
&=
    \frac{1}{\alpha} \cdot \sum_{k \in [K]} \frac{p_k^\alpha + (\alpha-1)q_k^\alpha-\alpha p_kq_k^{\alpha-1}}{\alpha-1}
\\
&=
    \sum_{k \in [K]} \frac{\lrb{\frac{p_k}{q_k}}^\alpha + (\alpha-1)-\alpha\frac{p_k}{q_k}}{\alpha(\alpha-1)} \cdot q_k^{\alpha}
\eqqcolon
    \sum_{k \in [K]} g_\alpha \lrb{\frac{p_k}{q_k}}\cdot q_k^{\alpha}
\eqqcolon
    (\square)\;,
\end{align*}
where
\[
    g_\alpha \colon [0,+\infty) \to \R\;, \qquad x \mapsto \frac{x^\alpha+ \alpha-1 -\alpha x }{\alpha(\alpha-1)}\;.
\]
Now, notice that $g_\alpha$ is a non-negative function (its minimum is at $1$ and its value is $0$) and hence
\begin{align*}
    (\square)
&\ge
    \sum_{k \in [K]} g_\alpha \lrb{\frac{p_k}{q_k}}\cdot q_k\;, \qquad \textit{if $\alpha \in (-\infty,0)\cup(0,1)$}\;,
\\
    (\square)
&\le
    \sum_{k \in [K]} g_\alpha \lrb{\frac{p_k}{q_k}}\cdot q_k\;, \qquad \textit{if $\alpha \in (1,\infty)$}\;.
\end{align*}
Now, we have
\begin{align*}  
    \sum_{k \in [K]} g_\alpha \lrb{\frac{p_k}{q_k}}\cdot q_k
&=
    \sum_{k \in [K]} \frac{\lrb{\frac{p_k}{q_k}}^\alpha + (\alpha-1)-\alpha\frac{p_k}{q_k}}{\alpha(\alpha-1)}\cdot q_k
\\
&=
    \sum_{k \in [K]} \frac{\lrb{\frac{p_k}{q_k}}^\alpha -1}{\alpha(\alpha-1)} \cdot q_k
=
    \sum_{k \in [K]} f_\alpha \lrb{\frac{p_k}{q_k}}\cdot q_k
=
    D^{\textrm{TRE}}_\alpha(p \Vert q) \;,
\end{align*}
and hence
\begin{align}
    \label{eq:bregman_vs_relative_a<1}
    D_\alpha(p \Vert q)
    &\ge
    D^{\textrm{TRE}}_\alpha(p \Vert q)\;, \qquad \textit{if $\alpha \in (-\infty,0)\cup (0,1)$\;,}
    \\
    \label{eq:bregman_vs_relative_a>1}
    D_\alpha(p \Vert q)
    &\le
    D^{\textrm{TRE}}_\alpha(p \Vert q)\;, \qquad \textit{if $\alpha \in (1,\infty)$\;.}
\end{align}
In particular, for $\alpha \in (-\infty,0)\cup(0,1)$, by combining \Cref{eq:bregman_vs_relative_a<1} and the generalized Pinsker inequality for $\alpha$-Tsallis relative entropies \Cref{eq:pinsker_TRE}, we get the following \emph{looser} form (by a $2^{1-\alpha}$ factor) of \Cref{thm:pinsker} for $\alpha \in (-\infty,0)\cup(0,1)$ and $p,q \in \relint(\simplex{K})$ 
\[
    D_\alpha(p\Vert q)
\ge
    D^{\textrm{TRE}}_\alpha(p \Vert q)
\ge
    \frac{1}{2} \cdot\lno{p-q}_1^2\;, \qquad \textrm{if $\alpha \in (0,1)$}\;.
\]
Since for $\alpha > 1$ and $p,q \in \relint(\simplex{K})$ the opposite inequality (Equation~\ref{eq:bregman_vs_relative_a>1}) holds, from \cref{thm:pinsker} we obtain the following \emph{looser} version of the generalized Pinsker inequality for $\alpha$-Tsallis relative entropies  
\[
    D^{\textrm{TRE}}_\alpha(p \Vert q)
\ge
    D_\alpha(p\Vert q)
\ge
    \frac{C_{\alpha,K}}{2} \cdot \lno{p-q}_1^2\;, \qquad \textrm{if $\alpha \in (1,+\infty)$}\;.
\]
The case $\alpha = 0$ is similar to the case $\alpha \in (-\infty,0)\cup (0,1)$, where the role of the $\alpha$-Tsallis relative entropy is played by the reverse KL, $(p,q) \mapsto \KL(q\Vert p)$.
Specifically, for any $p,q \in \relint(\simplex{K})$, we can obtain a \emph{looser} form (by a factor of $2$) of \Cref{thm:pinsker} as follows
\begin{align*}    
    D_0(p\Vert q)
&=
    \sum_{k \in [K]}\underbrace{\lrb{ \frac{p_k}{q_k}-\ln \frac{p_k}{q_k} -1}}_{\ge 0}
\ge
    \sum_{k \in [K]}\lrb{ \frac{p_k}{q_k}-\ln \frac{p_k}{q_k} -1}\cdot q_k
\\
&=
    \sum_{k \in [K]}\lrb{ -q_k\cdot\ln \frac{p_k}{q_k} }
=
    \sum_{k \in [K]}q_k\cdot\ln \frac{q_k}{p_k}
=
    \KL(q\Vert p)
\ge
    \frac{1}{2}\cdot\lno{p-q}_1^2\;,
\end{align*}
where the last inequality is the classic Pinsker inequality.
\section{Weaker Inequalities for \texorpdfstring{$\alpha > 2$ and $K \ge 3$}{alpha > 2 and K >= 3}}
\label{app:clipping}
As we have seen in Theorem~\ref{thm:pinsker} and in Appendix~\ref{app:nopinsker}, for $\alpha > 2$ and $K \geq 3$ no positive constant $C$ satisfies
$D_\alpha(p \Vert q) \ge C \cdot \lno{ p-q }_1^2$
for every 
$p$, $q \in \relint{\lrb{\simplex{K}}}$.
However, we have seen in \eqref{eq:change-of-variables} that
\begin{equation}\label{eq:useful}
    2\cdot \frac{D_\alpha\brb{p\Vert q}}{\lno{p - q}_1^2}
=
    \sum_{k \in [K]} v_k^2 \gamma_k^{\alpha -2}\;,
\end{equation}
where $\gamma$ is a convex combination of $p$ and $q$ and $v \coloneq \frac{p - q}{\lno{p - q}_1}$ is an element of $\tsphere{K}$.
In some applications (notably sequential prediction under logarithmic loss and related proper scoring rules),
it is standard to restrict predictions to a proper subset of the relative interior of the simplex or to apply truncation/smoothing so that
the loss and KL-type divergences remain finite and well-behaved
\citep{cesa-bianchi-lugosi2006plg,cesa-bianchi-lugosi2001logloss,gneiting-raftery2007jasa,vanommen2016robust}.
For instance, we may require that some positive $\e$ satisfying $0 < \e < \frac{1}{K}$ is such that $p_k \ge \e$  and $q_k \ge \e$ for every $k$, a condition which bind both $p$, $q$ to stay sufficiently away from the relative boundary of $\simplex{K}$.
This condition entails that also $\gamma_k \ge \e$ for every $k$, hence
$\sum_{k \in [K]} v_k^2 \gamma_k^{\alpha -2} \ge \e^{\alpha - 2} \lno{v}_2^2$.
Using that $\lno{v}_2^2 \ge C_{2, K}$ we derive
\begin{equation}\label{eq:relaxed-pinsk-both}
    D_\alpha\brb{p\Vert q}
\ge
    \frac{\e^{\alpha - 2}}{2}
\cdot
    C_{2, K}
\cdot
    {\lno{p - q}_1^2}\;.
\end{equation}
We may also derive a similar inequality in a less stringent setting, where only one among $p$ and $q$ is required to have its coordinates bounded below by $\e$. Indeed, the proof of Proposition~\ref{prop:second_order} shows that \eqref{eq:useful} is a consequence of the first-order Taylor expansion of $S_{\alpha}$ with Lagrange remainder. If instead we consider the integral remainder, which underlies the approach used, e.g., in \cite{reid2009surrogate,reid2011information} to estimate Bregman divergences, what we get is 
\begin{equation}\notag
    2\cdot \frac{D_\alpha\brb{p\Vert q}}{\lno{p - q}_1^2}
=
    \sum_{k \in [K]} v_k^2 \int_0^1 2(1-x) \brb{(1-x)q_k+xp_k}^{\alpha-2} \diff x\;,
\end{equation}
If the constraint is valid only for $p$, then for every $k$ we get
\begin{align*}    
    \int_0^1 2(1-x) \brb{(1-x)q_k+xp_k}^{\alpha-2} \diff x
&\ge
    \int_0^1 2(1-x) \brb{x\e}^{\alpha-2} \diff x
\\
&=
    \e^{\alpha - 2} \int_0^1 2(1-x)x^{\alpha-2} \diff x
=
    {\e}^{\alpha - 2} \cdot \frac{2}{\alpha (\alpha - 1)}\;.
\end{align*}
This results in the inequality
\begin{equation}\label{eq:relaxed-pinsk-p-constraint}
    D_\alpha\brb{p\Vert q}
\ge
    \frac{\e^{\alpha - 2}}{2}
\cdot
     \frac{2}{\alpha (\alpha - 1)}
\cdot
    C_{2, K}
\cdot
    {\lno{p - q}_1^2}\;.
\end{equation}
Similarly, if the constraint is valid only for $q$, then this time we have
\begin{align*}   
    \int_0^1 2(1-x) \brb{(1-x)q_k+xp_k}^{\alpha-2} \diff x\;
&\ge
    \int_0^1 2(1-x) \brb{(1-x)\e }^{\alpha-2} \diff x\;
\\
&=
    \e^{\alpha - 2} \int_0^1 2\brb{1-x}^{\alpha-1} \diff x\;
=
    \e^{\alpha - 2} \cdot \frac{2}{\alpha}\;,
\end{align*}    
hence
\begin{equation}\label{eq:relaxed-pinsk-q-constraint}
    D_\alpha\brb{p\Vert q}
\ge
    \frac{\e^{\alpha - 2}}{2}
\cdot
     \frac{2}{\alpha}
\cdot
    C_{2, K}
\cdot
    {\lno{p - q}_1^2}\;.
\end{equation}
We can rephrase what we have seen so far as follows: 
\begin{proposition}
    For $\alpha > 2$ and $K \geq 3$, consider the inequality
\begin{equation}\label{eq:relaxed-pinsk-generic}
    D_\alpha\brb{p\Vert q}
\ge
    \frac{C}{2}
\cdot
    {\lno{p - q}_1^2}\;.
\end{equation}
Then:
\begin{enumerate}
    \item For every $p$, $q \in \relint{\lrb{\simplex{K}}}$, the inequality \eqref{eq:relaxed-pinsk-generic} is satisfied by 
    \begin{equation}\notag
        C 
    =
        C(\alpha, K, p, q) 
    =  
        C_{2, K} 
    \cdot
        \lrb{\min_{k \in [K]}\lcb{p_k, q_k}}^{\alpha - 2}\;;
    \end{equation}
    \item For every $p$, $q \in \relint{\lrb{\simplex{K}}}$, the inequality \eqref{eq:relaxed-pinsk-generic} is satisfied by
    \begin{equation}\notag
        C 
    =
        C(\alpha, K, p) 
    =  
        C_{2, K} 
    \cdot
        \frac{2}{\alpha (\alpha - 1)}
    \cdot
        \lrb{\min_{k \in [K]}\lcb{p_k}}^{\alpha - 2}\;;
    \end{equation}
    \item For every $p$, $q \in \relint{\lrb{\simplex{K}}}$, the inequality \eqref{eq:relaxed-pinsk-generic} is satisfied by
    \begin{equation}\notag
        C 
    =
        C(\alpha, K, q) 
    =  
        C_{2, K} 
    \cdot
        \frac{2}{\alpha}
    \cdot
        \lrb{\min_{k \in [K]}\lcb{q_k} }^{\alpha - 2}\;.
    \end{equation}
\end{enumerate}
\end{proposition}
\end{document}